\documentclass[sigconf]{acmart}

\AtBeginDocument{%
  \providecommand\BibTeX{{%
    \normalfont B\kern-0.5em{\scshape i\kern-0.25em b}\kern-0.8em\TeX}}}

\copyrightyear{2021}
\acmYear{2021}
\setcopyright{acmcopyright}
\acmConference[SIGMOD '21]{Proceedings of the 2021 International Conference on Management of Data}{June 20--25, 2021}{Virtual Event, China}
\acmBooktitle{Proceedings of the 2021 International Conference on Management of Data (SIGMOD '21), June 20--25, 2021, Virtual Event, China}
\acmPrice{15.00}
\acmDOI{10.1145/3448016.3458456}
\acmISBN{978-1-4503-8343-1/21/06}

\usepackage[utf8]{inputenc}
\usepackage{xspace}
\usepackage{balance}
\usepackage{thmtools}
\usepackage{amsmath}
\usepackage{graphicx}
\PassOptionsToPackage{hyphens}{url}
\usepackage{hyperref}
\usepackage[hyphens]{url}
\usepackage{booktabs}
\usepackage{subcaption}
\usepackage{bm}
\usepackage[normalem]{ulem}
\usepackage{natbib}
\usepackage{enumitem}

\newtheorem{theorem}{Theorem}
\newtheorem{lemma}{Lemma}
\newtheorem{definition}{Definition}

\usepackage{dirtytalk}

\newcommand{\myparagraph}[1] { \noindent {\textbf {#1}} }
\newcommand{\paragemph}[1] { \noindent {\emph{#1}} }

\newcommand{\papertitle}{Correlation Sketches for Approximate Join-Correlation Queries}

\newcommand{\nycdataset}{NYC Open Data\xspace}
\newcommand{\nycdatasetshort}{NYC\xspace}
\newcommand{\wbfdataset}{World Bank Finances\xspace}
\newcommand{\wbfdatasetshort}{WBF\xspace}
\newcommand{\sbndataset}{Synthetic Bivariate Normal\xspace}
\newcommand{\sbndatasetshort}{SBN\xspace}

\newcommand{\R}{\mathbb{R}}
\newcommand{\dataset}{dataset\xspace}
\newcommand{\Dataset}{Dataset\xspace}
\newcommand{\datasets}{datasets\xspace}
\newcommand{\Datasets}{Datasets\xspace}
\newcommand{\corrsketch}{\textit{Correlation Sketches}\xspace}
\newcommand{\corrsketchnormal}{Correlation Sketches\xspace}

\newcommand{\hide}[1]{}

\renewcommand{\paragraph}[1]{\vspace{.08cm} \noindent \textbf{#1.}}

\newlist{myitemize}{itemize}{3}
\setlist[myitemize]{leftmargin=4mm}
\setlist[myitemize,1]{label=\textbullet}

\def\withcolors{1}
\def\withnotes{1} 

\ifnum\withnotes=1
  \usepackage[colorinlistoftodos,textsize=scriptsize]{todonotes}
\fi

\definecolor{Maroon}{rgb}{0.62, 0.0, 0.09}

\ifnum\withcolors=1
  \newcommand{\jfcolor}[1]{{\color{red}#1}} 
  \newcommand{\cmcolor}[1]{{\color{blue}#1}} 
  \newcommand{\ascolor}[1]{{\color{violet}#1}} 
  \newcommand{\abcolor}[1]{{\color{Maroon}#1}} 
\else
  \newcommand{\jfcolor}[1]{{#1}}
  \newcommand{\cmcolor}[1]{{#1}}
  \newcommand{\ascolor}[1]{{#1}}
  \newcommand{\abcolor}[1]{{#1}}
\fi

\ifnum\withnotes=1
  \newcommand{\asnote}[1]{\ascolor{\textbf{AS: }\sf #1}}
  \newcommand{\as}[1]{\ascolor{\textbf{Aécio: }\sf #1}}
  \newcommand{\cm}[1]{\cmcolor{\textbf{Chris: }\sf #1}}
  \newcommand{\jf}[1]{\jfcolor{\textbf{Juliana: }\sf #1}}
  \newcommand{\ab}[1]{\abcolor{\textbf{Aline: }\sf #1}}
\else
  \newcommand{\asnote}[1]{}
  \newcommand{\as}[1]{}
  \newcommand{\cm}[1]{}
  \newcommand{\jf}[1]{}
  \newcommand{\ab}[1]{}
\fi

\definecolor{HighlightColor}{rgb}{0.05,0.05,0.70}

\newcommand{\ignore}[1]{\leavevmode\unskip} 


\begin{document}

\fancyhead{}
\title{\papertitle}

\author{Aécio Santos$^1$, \ Aline Bessa$^1$, \ Fernando Chirigati$^2$, \ Christopher Musco$^1$, \ Juliana Freire$^1$}
\affiliation{
  \institution{$^1$New York University}
}
\email{{aecio.santos,aline.bessa,cmusco,juliana.freire}@nyu.edu}
\affiliation{
  \institution{$^2$Springer Nature}
}
\email{fernando.chirigati@us.nature.com}

\renewcommand{\shortauthors}{Aécio Santos, \ Aline Bessa, \ Fernando Chirigati, \ Christopher Musco, \ Juliana Freire}
\renewcommand{\authors}{Aécio Santos, \ Aline Bessa, \ Fernando Chirigati, \ Christopher Musco, \ Juliana Freire}

\begin{abstract}
  The increasing availability of structured datasets, from Web tables and open-data portals to enterprise data, opens up opportunities~to enrich analytics and improve machine learning models through relational data augmentation.
In this paper, we introduce a new class of data augmentation queries: \textit{join-correlation queries}.
Given a column $Q$ and a join column $K_Q$ from a query table $\mathcal{T}_Q$, retrieve tables $\mathcal{T}_X$ in a dataset collection such that $\mathcal{T}_X$ is joinable with $\mathcal{T}_Q$ on $K_Q$ and there is a column $C \in \mathcal{T}_X$ such that $Q$ is correlated with~$C$. 
A na\"ive approach to evaluate these queries, which first finds joinable tables and then explicitly joins and computes correlations between $Q$ and all columns of the discovered tables, is prohibitively expensive.
To efficiently support correlated column discovery, we
1)~propose a sketching method that enables the construction of an index for a large number of tables and that provides accurate estimates for join-correlation queries, and 2)~explore different scoring strategies that effectively rank the query results based on how well the columns are correlated with the query.
We carry out a detailed experimental evaluation, using both synthetic and real data, which shows that our sketches attain high accuracy and the scoring strategies lead to high-quality rankings.

\end{abstract}

\maketitle

\section{Introduction}
\label{sec:intro}

The increasing availability of structured \datasets,
from Web tables~\cite{cafarella@vldb2008,lehmberg@www2016} and open-data portals (see e.g.,~\cite{wbopendata,nycopendata,usopendata}) to enterprise data,
opens up opportunities to enrich analytics and improve
machine learning models through \emph{relational data augmentation}.
The challenge lies in finding relevant \datasets for a given task.
In this paper, we explore a new class of data search queries that retrieve related \datasets by uncovering correlations between numerical columns in unjoined \datasets.
Consider the following examples.

\paragemph{Example 1: Explaining Traffic Fatalities.} Based on the observation that there is a relationship between high traffic speed and the number of fatalities, the City of New York, as part of the Vision Zero initiative~\cite{vision-zero}, reduced the speed limit on its streets from 30 to 25 miles per hour.
While the lower speed limit was partially effective, fatalities
were still high and there was a need to understand why. 
Analysts 
hypothesized
that an increased number of bikes and the 
weather can also lead to more traffic fatalities.
By examining
\datasets published in NYC Open Data~\cite{nycopendata}, they validated their hypotheses:
from
CitiBike data, they found that when there are more active bikes, there
are more accidents; and they also observed that high
precipitation is correlated with a larger number of accidents.

\paragemph{Example 2: Improving Taxi Demand Models.}  A data scientist at the NYC Taxi \&
Limousine Commission 
created a
model to predict taxi demand
using a \dataset that contains historical data about taxi trips and
their associated pick-up times and location zip codes.
To improve the model, 
she had to find additional features that can influence taxi demand.
Domain experts suggested
a set of indicators for her to examine, including weather, major
events, and holidays. She obtained the relevant \datasets and observed  a substantial reduction in the root mean squared error after
augmenting
the training data.

In both examples, the analysts knew which \datasets they needed and where to find them.
Often, this is not the case: finding relevant data is a difficult and time-consuming task~\cite{fernandez@icde2018}.
On the Web, data are distributed over many sites and repositories, and in enterprises they are stored in a plethora of systems and databases.
As a point of reference, internal research at Lyft has found that their data scientists spend
25\% of their time on data discovery -- they only spend more time (40\%) on model development and deployment~\cite{lyft-amundsen}.

The interest in data search is growing in both industry and academia. In enterprises, there is an increasing number of ``data catalog'' systems~\cite{lyft-amundsen, linkedin-datahub,google-data-catalog, airbnb-dataportal} that support data search across multiple systems and data lakes. This market is projected to grow from USD 528.3 million in 2019 to USD 5.46 billion in 2027~\cite{data-catalog-market}.
\Dataset search engines that provide search over datasets on the Web~\cite{noy@www2019} are also becoming available.
While these systems fill important gaps in the data discovery space, they often have limited query capabilities,
supporting only queries over \dataset metadata.

Recent research proposes methods that support \dataset-oriented queries to retrieve \datasets that can be concatenated~\cite{nargesian@vldb2018} or joined with a given \dataset~\cite{zhu@sigmod2019,fernandez@icde2019,yang@icde2019}.
However, neither 
supports the discovery tasks illustrated in the examples above. Specifically, these tasks demand a new type of augmentation query to \emph{find \datasets that can be joined with and that also contains attributes that are correlated with those of a given query \dataset}.

\paragraph{Discovering Correlated \Datasets}
Augmentation queries over 
large \dataset collections can enable
analysts to uncover previously unknown relationships that lead to new hypotheses.
In Example 1, to explore the question ``What causes traffic fatalities?'', the analyst could search the NYC Open Data~\cite{nycopendata} repository for data \emph{related} to the traffic fatalities 
\dataset, which contains the total number of daily fatalities for each zip code; and in Example 2,  she could search
for data related to the taxi demand 
\dataset, which consists of the total number of hourly pickups per zip code.

For both examples, methods devised to find joinable tables~\cite{fernandez@icde2019,zhu@vldb2016-lsh-ensemble,zhu@sigmod2019} would  find relevant
results among over 2,000+ \datasets in NYC Open Data, but they would find \emph{too many}.
Specifically, these methods find columns that have a large overlap with the input \dataset, so they would return all \datasets that contain a column with zip codes in NYC -- most of which are not helpful to the users' information needs. For the examples above, further filtering is needed. To express their information need  more precisely, analysts would benefit from asking \emph{join-correlation queries}.

\begin{definition}[\bf Join-Correlation Query]
Given a column $Q$~and~a join column $K_Q$ from a query table $\mathcal{T}_Q$, a \emph{join-correlation query} finds tables $\mathcal{T}_X$ in a dataset collection  such that $\mathcal{T}_X$ is joinable with $\mathcal{T}_Q$ on $K_Q$ and there is a column $C \in \mathcal{T}_X$ such that $Q$ is correlated with $C$.
\end{definition}
\noindent In our first example, the analysts could issue a query to retrieve \datasets that join with the traffic fatalities \dataset and that also contain a column that correlates with the actual number of fatalities; in the second, the data scientist could search for \datasets that join with the taxi demand \dataset and that contain a column that correlates with the actual taxi demand.

%
%
The key challenge in evaluating \emph{join-correlation queries} is how to do so efficiently.
One approach is to first find joinable tables, and then to explicitly compute correlations between $Q$ and all columns of the discovered tables using, for instance, the Pearson's correlation coefficient for linear relationships, or the Spearman's rank correlation coefficient to capture non-linear relationships~\cite{CRR73}.
However, this approach requires the joinable \datasets returned to be downloaded and joined with the query table.
When these tables are large, 
they may not fit in memory and the cost of executing join operations can be prohibitive. Furthermore, some correlation measures are expensive to compute, e.g., to compute Spearman's correlation the data must first be sorted. 
This problem is compounded for queries that return a large number of
\datasets and require many joins and correlation computations to be
performed.
As a point of comparison, joining a \dataset on taxi pickups (approximately 1GB)  with a \dataset on precipitation (approximately 3MB) took about 29 seconds
and computing the Spearman's coefficient between the numbers of pickups and precipitation levels in the resulting joined \dataset took about 5 seconds (on a Intel Core i5 2.4GHz CPU).

\paragraph{Our Approach: Ranking \Datasets via Correlation Estimates}
To reduce the evaluation cost of join-correlation queries, as an alternative to using the entire data, we investigate the use of data synopses to estimate the results.
Data-intensive algorithms can often be optimized by reducing the size of the input data with sampling techniques~\cite{cormod@foundtrends2012}, at the cost of obtaining approximate results. In our setting, however,  na\"ively sampling data to estimate correlations does not work: it is not possible
to sample columns to estimate the correlation without first executing the join.
To support the \textit{efficient discovery of correlated columns from distinct \datasets at scale},
we propose a sketching method for \emph{estimating correlations between columns from unjoined \datasets} based on column synopses, namely \corrsketch.
%
These synopses are constructed using only data from individual columns, and thus they can be \emph{pre-computed and indexed} to support discovery of joinable \datasets and \emph{fast correlation estimation} at query time.

Our method constructs a synopsis $L_{\langle K_X, X \rangle}$
for any given pair of columns $\langle K_X, X \rangle$ that belongs to a table $\mathcal{T}_X$, where $K_X$ is a categorical column and $X$ is a numerical column. A pair of synopses $L_{\langle K_X, X \rangle}$  and $L_{\langle K_Y, Y \rangle}$ (for tables $\mathcal{T}_X$ and $\mathcal{T}_Y$ respectively) can be used to estimate the correlation between the numerical columns $X_{X \bowtie Y}$ and $Y_{X \bowtie Y}$,  generated \emph{after} joining tables $\mathcal{T}_X$  and $\mathcal{T}_Y$ on columns $K_X$ and $K_Y$. Note that $\mathcal{T}_X$ and $\mathcal{T}_Y$ are heterogeneous and need not have the same join keys or the same number of rows.
As such, our sketching method enables the construction of an index for a large number of tables that can be used to support \emph{both} joinability queries
and to estimate the correlation between a query column and indexed columns.

Our sketching method builds upon and extends state-of-the-art hashing techniques~\cite{bar-yossef@random2002, beyer@cacm2009, beyer@sigmod2007, yang@icde2019, huang2019joins}. In particular, we 
prove that our sketches can reconstruct a uniform random sample of the paired columns $X_{X \bowtie Y}$ and $Y_{X \bowtie Y}$. This leads to an important property: besides correlations, our sketching approach can handle \emph{any} statistic that can be estimated from random samples (e.g., entropy and mutual information). While other statistics are out of the scope of this paper, we demonstrate the flexibility of our method by using it to estimate a variety of different correlation coefficients.

We show both theoretically and experimentally that our approach is effective and provides accurate estimates for correlation.
Moreover, our analysis provides mathematical tools for dealing with approximation errors typical of sketching algorithms.
For join-correlation queries, these errors may lead to false positives: columns are returned which seem more correlated, based on the sketch, than they actually are.
This problem is an issue for queries over large \dataset collections, as there can be many false positives, simply by chance. 
To address this issue, we derive sub-sample confidence interval bounds to estimate approximation errors.
Our bounds are based on simple columns statistics like sample size and data range and, in contrast to  prior work~\cite{bishara@bjmsp2018, bishara@brm2017, berry2000monte}, do not rely on distributional assumptions.
We use these bounds to design a set of scoring functions that rank \datasets based on both their estimated correlation with a query \dataset, and on our confidence in that estimate.

We evaluate our method experimentally using both synthetic and real-world \datasets.
A comparison of the estimates produced by \corrsketch with the actual correlation values shows that it derives accurate estimates.
In addition, we assess the effectiveness of different ranking functions that leverage
\corrsketch, and show that they improve ranking performance up to 193\% in terms of mean average precision when
compared to a scoring scheme based on overlap size, commonly used for joinability queries~\cite{fernandez@icde2019,zhu@vldb2016-lsh-ensemble,zhu@sigmod2019}.

%
\paragraph{Contributions}
We introduce \textit{join-correlation queries}, a new class of queries to find correlated \datasets within a collection of disconnected tables, and
propose new methods to efficiently support such queries over large \dataset collections. To the best of our knowledge, ours is the first work that addresses this problem.
Our contributions can be summarized as follows:
\vspace{-.1cm}
\begin{myitemize}
\item We propose \corrsketch, a new sketch that simultaneously summarizes information about joinability and correlation, allowing the
estimation of different correlations measures between columns of unjoined \datasets
(Section~\ref{sec:correlation-sketches}).
\item
We derive new correlation confidence interval bounds that allow us to measure the risk of estimation errors. These bounds serve as the basis for the design of scoring functions that use our sketches to rank the discovered columns (Section~\ref{sec:ranking}).
\item We perform an extensive experimental evaluation of our method and show that: \corrsketch estimates correlations with good accuracy in both synthetic and real data for different correlation measures; and the scoring functions we propose are effective and derive high-quality rankings (Section~\ref{sec:experiments}).
\end{myitemize}

\section{Preliminaries}
\label{sec:preliminaries}
Our approach extends existing hashing-based methods proposed for cardinality estimation to the join-correlation query problem. We describe these methods below.
For more details, we refer the reader to Section~\ref{sec:related} and to a comprehensive survey on the topic~\cite{cormod@foundtrends2012}. We also review the literature on correlation estimation and discuss the properties of correlation estimators that we use in this paper.

\vspace{-.3cm}
\subsection{Cardinality Estimation via Sketches}

\paragraph{Estimating Distinct Values}
The problem of determining the number of distinct values (DV) in a \dataset has been extensively studied~\cite{brown@datamgmt,ioannidis@vldb2003,padmanabhan@sigmod2003}
Since computing the \textit{exact} number of distinct elements 
is expensive,
approximate methods have been proposed that can scale to massive collections of  \datasets.
Effective approaches for DV estimation rely on hashing techniques, require a single pass through the data, and use a bounded amount of memory~\cite{harmouch@vldb2017cardest-survey}.

Let $h_u$ be a hash function that maps distinct values randomly and uniformly to the unit interval $[0, 1]$,
and $D$ be the number of distinct elements in a \dataset.
The key idea behind DV estimators is that, if we use $h_u$ to map elements to the unit interval and the number of distinct elements in a \dataset is large (i.e., $D \gg 1$), then the expected distance between any 
two neighboring points in the unit interval is $1/(D + 1) \approx 1/D$, 
and the expected value of the $k^{th}$ smallest point, $U(k)$, is estimated with $\mathbb{E}[U(k)] \approx \sum_{j=1}^{k} (1/D) = k / D$. Thus, the number of distinct values in the \dataset can be approximated by $D \approx k / \mathbb{E}[U(k)]$. 
The simplest estimator of $\mathbb{E}[U(k)]$ is $U(k)$ itself, yielding the basic estimator:  $\hat{D}_k^{BE} =  k/U(k)$.

Based on this idea, algorithms and methods for building synopses have been developed to estimate set cardinality~\cite{beyer@cacm2009}.
An example is the popular \textit{$k$ Minimum Values} (KMV) synopsis (also known as bottom-$k$ sketches), introduced by 
Bar-Yossef et al.~\cite{bar-yossef@random2002}.
Concretely, a KMV synopsis of a set $X$ comprises the $k$ minimum hash values of the elements of $X$, generated by a hash function $h_u$ mapping to $[0, 1]$. To estimate the number of distinct elements $|X|$, one can use this synopsis and a 
DV estimator, such as $\hat{D}_k^{BE}$.
%
Alternatively, an improved DV estimator proposed by Beyer et.~al.~\cite{beyer@sigmod2007} can be used.
Their estimator, given by $\hat{D}_k^{UB} = (k-1)/U(k)$, is
unbiased, has a lower mean squared error, and has the minimal possible variance of any DV estimator when there are many distinct values and the synopsis size is large.

\paragraph{Cardinality Estimation under Set Operations}
Beyer et~al.~\cite{beyer@sigmod2007}  also considered 
how multiple KMV synopses, created independently, can be combined to estimate the cardinality of sets that result from multi-set operations (e.g., union, intersection, and difference).
As an example, consider a set $X$ composed of two partitions $X_A$ and $X_B$, i.e., $X = X_A \cup X_B$. Next, let  $L_A$ and $L_B$ be the KMV synopses of sets $X_A$ and $X_B$, with sizes $k_A$ and $k_B$  respectively. We can combine $L_A$ and $L_B$ to build a valid KMV synopsis $L = L_A \oplus L_{B}$, where $\oplus$ is an operator for combining two KMV synopses. 
$L$ represents the set comprising the $k$ smallest values in $L_A \cup L_B$, where $k = min(k_A, k_B)$. To estimate the number of distinct elements $D_{\cup}$ in the union $X_A \cup X_B$, we can directly use estimator $\hat{D}_k^{UB}$
on the synopsis $L$.
%
To estimate the number of distinct values $D_{\cap}$ in the intersection $X_A \cap X_B$, we must first compute the number of common distinct hashes in $L_A$ and $L_B$ (i.e., $K_{\cap} = |\{v \in L : v \in L_A \cap L_B\}|$). Then, we can estimate $D_{\cap}$ as:
\vspace{-.2cm}
\begin{equation}
\label{eq:intersection-estimator}
\widehat{D_{\cap}} = \frac{K_{\cap}}{k} \frac{k-1}{U(k)}
\vspace{-.1cm}
\end{equation}


\subsection{Correlation Estimation}
\label{sec:prelim_corr_estimation}

The problem of measuring dependence between a pair of vectors has been studied for over a  century~\cite{rodgers@tas1988}, and new correlation measures continue to be developed~\cite{baak2020new, szekely@tas2007}. Pearson's correlation coefficient is one of the oldest and most widely used correlation measures~\cite{CRR73}. While our methods can be used to estimate 
any measure of correlation, we use Pearson's as our main motivating example. 

When applied to a population, Pearson's correlation coefficient is usually referred to as $\rho$~\cite{CRR73}. 
For a pair of random variables $\langle X, Y \rangle$, the coefficient is defined as:
\begin{equation}
\label{eq:pearson}
\rho_{XY} = \frac{\mathbb{E}[ (X - \mu_{X}) (Y - \mu_{Y}) ]}{\sigma_{X}\sigma_{Y}}
\end{equation}
where $\sigma_X$ (resp. $\sigma_Y$) is the standard deviation of the random variable $X$ (resp. $Y$), and $\mu_X$ (resp. $\mu_Y$) is the mean of $X$  (resp. $Y$). $\rho_{XY}$ can be estimated with a finite sample from distributions $X$ and $Y$ using what is usually referred to as Pearson's sample correlation ($r$):
\begin{equation}
\label{eq:pearson-sample}
r_{XY} = \frac{\sum ^n _{i=1}(x_i - \bar{x})(y_i - \bar{y})}{\sqrt{\sum ^n _{i=1}(x_i - \bar{x})^2} \sqrt{\sum ^n _{i=1}(y_i - \bar{y})^2}}.
\end{equation}
Above $n$ is the sample size, $x_{i}$ and $y_{i}$ are individual
samples,
and $\bar{x}$ and $\bar{y}$ are, respectively, the means of the sub-samples of $X$ and $Y$.

There is a lot of prior work on understanding the accuracy of sample correlation estimators, but these works typically make strong data assumptions.
When data follows a bivariate normal distribution, the sampling distribution of $r$ is asymptotically normal and centered around $\rho$~\cite{shevlyakov@robust-corr-book}.  For a finite sample of size $n$, the variance of $r$ is known to depend both on $n$ and on the underlying population correlation $\rho$ ~\cite{shieh@brm2010}:
\vspace{-.2cm}
\begin{equation}
\label{eq:pearsons-variance}
Var(r) = \frac{ (1 - \rho^2)^2 }{ n - 1} 
\end{equation}

When data is not normally distributed (as is often the case in practice), less is known. Nevertheless, there has been an increasing interest in the non-normal setting in recent years ~\cite{yuan@jmva2000, dewinter@psymethods2016, bishara@psymethods2012, bishara-hittner@epm2015, bishara@brm2017, bishara@bjmsp2018, hu@tas2020}. 
Yuan and Bentler \cite{yuan@jmva2000, yuan2005effect} show asymptotically that the standard deviation of the sample estimator $r$ depends on the joint fourth-order moments, or kurtoses, of the variables. In agreement, empirical simulations confirm that the presence of a high excess kurtosis can lead to increased bias and estimator errors~\cite{dewinter@psymethods2016,bishara@brm2017}. 

\paragraph{Robustness and Alternative Correlations} 
One challenge in obtaining finite sample accuracy bounds, as we do in Section~\ref{sec:confidence-interval}, is that Pearson's correlation coefficient is known to be sensitive to outliers.
In fact, it has been shown that a single sample $(x_i, y_i)$ can have an unbounded effect on the correlation and 
can potentially lead to catastrophic estimation errors~\cite{devlin@biometrika1975}.
This fact has spurred the development of correlation estimators that are robust against outliers and distribution contamination~\cite{devlin@biometrika1975, shevlyakov@robust-corr-book}. However, these robust estimators are less efficient than $r$, i.e., require larger sample sizes. We refer the reader to \cite{shevlyakov@robust-corr-book} for an extensive review of robust correlation estimators.
Resampling-based approaches such as bootstrapping~\cite{efron1994introduction} can also be used to reduce error in estimating Pearson's correlation, especially at small sample sizes~\cite{bishara@brm2017}. These approaches, however, have a much higher computational cost~\cite{bishara@brm2017}.

Alternative correlation measures, such as the Spearman's rank correlation coefficient~\cite{CRR73}, and data distribution transformations, such as the Rank-based Inverse Normal (RIN)~\cite{bliss1967statistics, bishara-hittner@epm2015}, may be more effective than Pearson for highly non-normal data~\cite{bishara@brm2017}.
However, they have different semantics from Pearson's (e.g., they capture non-linear relationships), and thus the choice of correlation measure depends on the user's application.
All of these correlations can be estimated with our approach, and we experimentally show that the accuracy of our estimates for these estimators are similar in the data collections we have considered (Section~\ref{sec:experiments}).

\section{Estimating Join-Correlation}
\label{sec:correlation-sketches}

\begin{figure}[t]
\centering
\scriptsize
\fontfamily{phv}\selectfont
\parbox{.25\linewidth}{
\centering
    \begin{tabular}{cc}
    \multicolumn{2}{c}{$\bm{\mathcal{T}_X}$}        \\
    \hline
    \textbf{$K_X$} & \textbf{$X$} \\
    \hline
	2021-01  &  6.0 \\ 
	2021-02  &  4.0 \\ 
	2021-03  &  2.0 \\ 
	2021-04  &  3.0 \\ 
	2021-05  &  0.5 \\ 
	2021-06  &  4.0 \\ 
	2021-07  &  2.0 \\ 
    \hline
    \end{tabular}
}
\parbox{.25\linewidth}{
    \centering
    \begin{tabular}{cc}
    \multicolumn{2}{c}{$\bm{\mathcal{T}_Y}$}        \\
    \hline
    \textbf{$K_Y$} & \textbf{$Y$} \\
    \hline
	2021-01  &  5.5 \\ 
	2021-01  &  4.5 \\ 
	2021-02  &  3.9 \\ 
	2021-02  &  2.0 \\ 
	2021-03  &  4.0 \\ 
	2021-03  &  1.0 \\ 
	2021-04  &  4.0 \\ 
    \hline
    \end{tabular}
}
\parbox{.4\linewidth}{
    \centering
    \begin{tabular}{ccc}
    \multicolumn{3}{c}{$\bm{\mathcal{T}_{X \bowtie Y}}$}        \\
    \hline
    \textbf{$K_{X \bowtie Y}$} & \textbf{$X_{X \bowtie Y}$} & {$Y_{X \bowtie Y}$} \\
    \hline
	2021-04  &  3.0  &  4.0 \\ 
	2021-03  &  2.0  &  2.5 \\ 
	2021-02  &  4.0  &  3.0 \\ 
	2021-01  &  6.0  &  5.0 \\ 
    \hline \\ \\
    \end{tabular}
} 
\vspace{-.2cm}
\caption{
Table $\mathcal{T}_{X \bowtie Y}$ is the join of the input tables $\mathcal{T}_X$ and $\mathcal{T}_Y$, aggregated using the \texttt{mean} function.
\corrsketch efficiently reconstructs a sample of the table $\mathcal{T}_{X \bowtie Y}$ to estimate correlation between the columns $X_{X \bowtie Y}$ and $X_{X \bowtie Y}$, without computing the full join.  
}
\vspace{-.3cm}
\label{fig:example-tables}
\end{figure}

Before presenting \corrsketch, we introduce notation and formally define the join-correlation estimation problem.
%
Consider 
a query table $\mathcal{T}_X$ composed of a categorical column $K_X$
and a numerical column $X$, and a table $\mathcal{T}_Y$ 
in a dataset collection 
containing a categorical
column $K_Y$ and a numerical column $Y$
(we discuss below how multi-column tables are handled). 
Columns $K_X$ and $K_Y$ 
%
are  the join attributes, i.e., $\mathcal{T}_{X \bowtie Y} = \pi_{k, x_k, y_k} (\mathcal{T}_X\bowtie_{K_X=K_Y} \mathcal{T}_Y)  = \{ \langle k, x_k, y_k \rangle : k \in K_X \cap K_Y \}$.
We denote by $x_k$ and $y_k$ the numerical values of $X$ and $Y$ (respectively) associated with the row identified by key $k$.
This is illustrated in Figure~\ref{fig:example-tables}.

\begin{definition}[\bf Join-Correlation Estimation]
Given two tables $\mathcal{T}_X$ and $\mathcal{T}_Y$,
we aim to efficiently estimate the correlation $r_{X \bowtie Y}$ of the numerical attributes $X_{X \bowtie Y}$ and $Y_{X \bowtie Y}$  in $\mathcal{T}_{X \bowtie Y}$ without having to compute the join and aggregations for $\mathcal{T}_X$ and $\mathcal{T}_Y$.
\end{definition}

\noindent One  approach to estimate $r_{X \bowtie Y}$ is to use data sketches instead of  the full \datasets.  
To do so, we can build synopses $L_{\langle K_X,X \rangle}$ and $L_{\langle K_Y,Y \rangle}$ 
that serve as summaries of $\langle K_X, X \rangle$ and $\langle K_Y, Y \rangle$  respectively.
However, na\"{\i}ve approaches do not yield useful summaries.

\paragraph{Limitations of Random Sampling}
Consider, for example, two numerical vectors of size $n$, $S_X$ and $S_Y$, randomly sampled from the numerical columns $X \in \mathcal{T}_X$ and $Y \in \mathcal{T}_Y$, respectively. The correlation between $S_X$ and $S_Y$ is not a valid estimate of the correlation between $X_{X \bowtie Y}$ and $Y_{X \bowtie Y}$ because the pairs $\langle x \in S_X, y \in S_Y \rangle$ are not aligned.
By sampling directly from $X$ and $Y$, we lose information on what keys are associated with what numerical values, which is necessary  to  align pairs  $\langle x_k, y_k \rangle$. 
Another alternative 
would be to include the keys by randomly sampling rows from original tables $\mathcal{T}_X$ and $\mathcal{T}_Y$, and then joining the row samples.
However, since the final set of keys $k \in K_{X \bowtie Y}$ depends on input columns $K_X \in \mathcal{T}_X$ and $K_Y \in \mathcal{T}_Y$, it is unlikely that the keys selected from $K_X$ \textit{and} contained in $K_{X \bowtie Y}$ will also be selected from $K_Y$. We discuss this issue more formally in the next subsection.


\begin{figure}[t]
\scriptsize
\fontfamily{phv}\selectfont
\parbox{.35\linewidth}{
    \centering
    \begin{tabular}{ccc}
    \multicolumn{3}{c}{$\bm{L_{\langle K_X,X \rangle}}$} \vspace{2px} \\
    \hline
    $h(k)$ & $h_u(k)$ & $x_k$ \\
    \hline
	bac52e98    &  0.48  &  2.0 \\
	16dab449    &  0.34  &  2.0 \\
	26f79756    &  0.47  &  3.0 \\
	4da33cf5    &  0.34  &  6.0 \\
    \hline \\
    \end{tabular}
}
\parbox{.34\linewidth}{
    \centering
    \begin{tabular}{ccc}
    \multicolumn{3}{c}{$\bm{L_{\langle K_Y,Y \rangle}}$} \vspace{2px} \\
    \hline
    $h(k)$ & $h_u(k)$ & $y_k$ \\
    \hline
	16dab449    &  0.34  &  2.5 \\
	bd5a7c1f    &  0.89  &  3.0 \\
	26f79756    &  0.47  &  4.0 \\
	4da33cf5    &  0.34  &  5.0 \\
    \hline \\
    \end{tabular}
}
\parbox{.29\linewidth}{
    \centering
    \begin{tabular}{ccc}
    \multicolumn{3}{c}{$\bm{L_{\langle X \bowtie Y \rangle}}$} \vspace{2px} \\
    \hline
    \textbf{$h(k)$} & \textbf{$x_k$} & $y_k$ \\
    \hline
	16dab449    &  2.0 &  2.5 \\
	26f79756    &  3.0 &  4.0 \\
	4da33cf5    &  6.0 &  5.0 \\
    \hline \\
    \end{tabular}
}
\vspace{-.5cm}
\caption{
The tables $L_{\langle K_X,X \rangle}$ and $L_{\langle K_Y,Y \rangle}$ represent correlation sketches for the tables $\mathcal{T}_X$ and $\mathcal{T}_Y$, for sketch size $n = 3$ and \texttt{mean} aggregation.
While we explicitly show the column $h_u(k)$ for illustrative purposes, it does not need to be stored as it can be easily computed from $h(k)$.
}

\label{fig:example-sketches}
\vspace{-.5cm}
\end{figure}

\vspace{-.1cm}
\subsection{\corrsketchnormal}

\corrsketchnormal address the limitations of na\"ive approaches by \emph{enabling the reconstruction of a uniform random sample of the joined table}.
Our method uses hashing techniques to carefully select a small sample of tuples $\langle k, x_k \rangle$ from a table 
$\mathcal{T}_X = \langle K_X, X \rangle$, that is used to build a sketch which enables data from different tables to be aligned and correlations to be estimated.

\paragraph{Sketch Construction} We use two different hashing functions to create the sketch $L_{\langle K_X, X \rangle}$ for table $\mathcal{T}_X$. The first one, $h$, is a collision-free hash function that randomly and uniformly maps key values $k \in K_X$ into distinct integers. 
Given that these integers $h(k)$ are unique, they are used as the tuple identifiers in sketch  $L_{\langle K_X, X \rangle}$.
Next, we use hashing function $h_u$ to map integers $h(k)$ to real numbers in the range $[0, 1]$, uniformly at random. 
The function $h_u$ plays a key role in the selection of the $n$ tuples that compose the sketch $L_{\langle K_X, X \rangle}$: the tuples that correspond to the $n$ smallest $h_u$ values are the ones included in the sketch.
More formally, we select $n$ samples of pairs $\langle h(k), x_k \rangle$ with minimum values of $h_u(k)$, i.e., $L_{\langle K_X, X \rangle} = \{\langle h(k), x_k \rangle :  k \in min(k, h_u(k))\}$, where $min$ is a function that returns a set containing the keys $k$ with the $n$ smallest values of $h_u(k)$.
To illustrate this, consider the example table $\mathcal{T}_X$ in Figure~\ref{fig:example-tables}. 
To build a correlation sketch, we apply the hashing functions to each one of the keys $k \in K_X$ and then select the $n$ tuples associated~with the smallest hashed values, which are used to create the \textit{correlation sketch} $L_{\langle K_X, X \rangle}$ in Figure~\ref{fig:example-sketches}.

Once these sketches are created (independently) for separate tables, they are used to estimate correlation by computing a \emph{joined sketch} $L_{X \bowtie Y}$, also illustrated in Figure~\ref{fig:example-sketches}. $L_{X \bowtie Y}$ has a row for every key $k$ that appears in both $L_{\langle K_X, X \rangle}$ and $L_{\langle K_Y, X \rangle}$. As we argue in Theorem \ref{thm:uniform}, this table contains a uniform random sample of paired numerical values from $\mathcal{T}_{X \bowtie Y}$, so it can be used directly to estimate correlation or any other statistic over $\mathcal{T}_{X \bowtie Y}$.

The accuracy of an estimate, however, depends on the size of this sample being large, i.e., on $L_{X \bowtie Y}$ having many rows. The key idea behind our method is that, 
by selecting samples using $h_u$, we introduce dependence that increases the probability of $L_{\langle K_X, X \rangle}$ and $L_{\langle K_Y, X \rangle}$ including the same keys \cite{huang2019joins}.
To see this, consider an extreme example where both $K_X$ and~$K_Y$ have the same set of $N$ distinct keys.
Suppose that $n$ key-value tuples $\langle k, x_k \rangle$ are included in $L_{\langle K_X, X \rangle}$ uniformly at random (without hashing), and $n$ are also included in $L_{\langle K_Y, Y \rangle}$, independently and uniformly at random.
If a key $k$ is included in $L_{\langle K_X, X \rangle}$, the probability that it also appears in $L_{\langle K_Y, Y \rangle}$ is $n/N$.
Thus, the expected number of rows in $L_{X \bowtie Y}$ is $n^2/N$, vanishingly small when the sketch size $n$ is much smaller than $N$. With little or no key overlap, we have no way of effectively estimating correlation from $L_{X \bowtie Y}$.
On the other hand, when the inclusion is determined by the values of $h_u(k)$, the events become dependent: if $k$ is included in $L_{\langle K_X, X \rangle}$, $k$ must also be included in $L_{\langle K_Y, Y \rangle}$.
In this case, the number rows in $L_{X \bowtie Y}$ increases to $n$, the maximum number possible.
In a less extreme case where $K_X$ and~$K_Y$ do not have the exact same keys, the expected number of keys included in both sketches will depend on the Jaccard similarity between the key sets, but in any case, will be much larger than the $n^2/N$ obtained by uniform random sampling.



Note that the use of a hashing function such as $h_u$ to introduce dependence is not a new idea -- it has been used in many algorithms~\cite{beyer@cacm2009, vengerov@vldb2015, dasgupta@icdt2016-theta-sketches, huang2019joins, yang@icde2019}. Part of our contribution lies in the combination of $h_u$ with another function $h$ to generate tuple identifiers that \emph{allow the alignment of paired data samples at estimation time}.
\corrsketchnormal contains both the $n$ minimum hashed values $h(k)$ and their corresponding numerical values $x_k$ from column $X$. 
By keeping the hash of the key, it is possible to
align the numerical values with values in other tables that are associated
with the same key, and by storing the numerical values, we can estimate the correlations between the numeric columns.

\paragraph{Handling Repeated Keys}
The process described above assumes that the keys uniquely identify each row in a table.
However, real-world data often contain repeated categorical values 
(as in column $K_Y$ in Figure \ref{fig:example-tables}).
In such cases,  there is a set of values associated with each distinct key $k$. For instance, in Figure~\ref{fig:example-tables}, the set of values $\{5.5, 4.5\}$ is associated with the key ``\texttt{2021-01}''.
Because correlation is only defined for sets of paired values, downstream applications that use correlation typically aggregate the numeric values associated with a key into a single number. This can be done by applying  a user-defined function (e.g., \textit{mean}, \textit{sum}, \textit{maximum}, \textit{minimum}, \textit{first}, \textit{last}) to perform the aggregation before computing the  correlation.  
In Figure~\ref{fig:example-tables},
the column  $Y_{X \bowtie Y}$ contains the aggregated values of $Y$ using the
\textit{mean} function after the join between $\mathcal{T}_X$ and $\mathcal{T}_Y$.

Repeated keys can be handled during sketch construction: 
whenever a key $k$ that already exists in the set of hashed minimum values is found again at time $t$, an aggregate function $f$ can be applied to compute the value for time $t$ by aggregating the existing $x_k^{t-1}$ with the new incoming $x_k$, i.e., $x_k^t = f(x_k, x_k^{t-1})$.
As long as the aggregation can be computed in a streaming fashion, the synopses can also be computed with a single pass over the data.
$L_{\langle K_Y, Y \rangle}$ in Figure~\ref{fig:example-sketches} shows a sketch  constructed from table $\mathcal{T}_Y$ in Figure~\ref{fig:example-tables}, using  \textit{mean} as the aggregate function $f$.

Note that the choice of function affects the semantics of the data, thus this selection must be made by taking into account the requirements of the downstream application that makes use of the sketches.
Nonetheless, our synopsis is agnostic to such aggregations, and can easily be extended to  take as input one or more functions.


\paragraph{Sketches for Multi-Column Tables}
For simplicity, we described how to build sketches from a binary table,
but it is trivial to extend to process to multi-column tables.
For example, if a table contains multiple columns, $\mathcal{T}_{XZ} = \{K_{XZ}, X, Z\}$, the correlation sketch could be extended to $L_{\langle K_{XZ}, X, Z \rangle} = \{\langle h(k), x_k, z_k \rangle :  k \in min(k, h_u(k))\}$.
Alternatively, one could simply build one sketch for each pair of keys and numeric columns, e.g., $\langle K_X, X \rangle$ and $\langle K_Z, Z \rangle$.

\vspace{-.1cm}
\subsection{Estimating Join-Correlation}
An important property of \corrsketch is that it enables the construction of a \emph{uniform random sample} of the join $\mathcal{T}_{X \bowtie Y}$.
This can be formally stated as (proof is deferred to the appendix): 
\begin{theorem}
\label{thm:uniform}
The set of paired numeric values $\langle x_k, y_k \rangle \in L_{X \bowtie Y}$ is a uniform random sample of the set of paired numeric values $\langle x_k, y_k \rangle \in \mathcal{T}_{X \bowtie Y}$. 
\end{theorem}

\noindent Common measures of correlation, such as
Pearson and Spearman, can be approximated with a sub-sample of data from those columns, as long as that sample is taken uniformly at random. 
Theorem~\ref{thm:uniform} forms the basis of our algorithm for join-correlation estimation, which consists of two steps:
(1) create the synopsis table $L_{X \bowtie Y }$ by performing a join between two synopses $L_{X}$ and $L_{Y}$ on their hashed keys $h(x)$ (as illustrated in Figure \ref{fig:example-sketches}), and (2) apply \emph{any} sample correlation estimator to the numerical data of $L_{X \bowtie Y }$ to estimate the correlation between columns $X$ and $Y$ in $\mathcal{T}_{X \bowtie Y}$.

\vspace{-.1cm}
\subsection{Discussion} \label{sec:corrsketches-discussion}
From Theorem \ref{thm:uniform}, we know that correlation sketches provide a valid estimate of the correlation between any two data columns \emph{after a join} -- i.e., between $X_{X\bowtie Y}$ and $Y_{X\bowtie Y}$.
While this estimate is often accurate, it is based on a sub-sample of data and inaccurate estimates are inevitable (see e.g., Figure~\ref{fig:scatterplots-all-datasets}). Of course, this will be true for any randomized estimator, not just our algorithm.


The variance of \corrsketchnormal estimates depends on the correlation estimator used. In general, as we show in Section \ref{sec:experiments}, correlation estimates converge to the true correlation when the sketch join sample size (i.e., the number of rows in $L_{\langle X \bowtie Y \rangle}$) increases. 
This sketch join size depends on multiple factors. First, there is a space-accuracy trade-off: as the number of minimum hash $n$ increases, the probability of having larger join sizes also increases. The sketch join size also depends on the distribution of the join keys.
Therefore, the hashing selection strategy used to include tuples in the synopsis affects the sketch intersection size, and ultimately it also affects the variance of correlation estimation.

The fixed-size sample selection strategy adopted in this paper is similar to the hashing strategy in~\cite{beyer@sigmod2007}, i.e., the sketch contains the $n$ minimum values of $h_u(k)$.
However, there is a wide range of possible hashing strategies with variable size \cite{dasgupta@icdt2016-theta-sketches, yang@icde2019} that may have different accuracy-space trade-offs and could be used in \corrsketchnormal. Exploring the effect of different selection strategies on join-correlation estimation is an open problem that we would like to explore in future work.

Another benefit of \corrsketchnormal is that it retains all information contained in a KMV sketch~\cite{bar-yossef@random2002, beyer@sigmod2007}.
Therefore, it allows not only for the estimation of correlations associated to numerical columns, but it also enables the estimation of  all statistics supported by the family of minimum-value sketches (e.g., cardinality, Jaccard containment, and similarity). These could be used, for example, to estimate the number of distinct elements in each individual column ($K_X$ and $K_Y$), the containment of $K_X$ in $K_Y$, and the size of the resulting join table $\mathcal{T}_{X \bowtie Y}$.

\subsection{Implementation}
We used the well-known 32-bits MurmurHash3 function to implement $h$, since it has been shown to perform similarly to truly random hashing functions~\cite{soren@neurips2017}. For $h_u$, we used Fibonacci hashing~\cite{knuth1997art}, a simple multiplicative hashing function (also known as the \textit{golden ratio multiplicative hashing}).
To build the synopses, we implemented a tree-based algorithm similar to the one described in~\cite{beyer@sigmod2007}.
In summary, the algorithm performs one pass in the data while maintaining a tree that keeps the $n$ tuples $\langle h(k), h_u(k), x_k \rangle$ with minimum $h_u(k)$ values.
As we discuss in Section~\ref{sec:exp-estimators}, we also implemented multiple correlation estimators.

\section{Ranking Correlated Columns}
\label{sec:ranking}

To query large dataset collections, we focus on a variation of join-correlation queries that retrieve the top-$k$ results,
which we define more formally as follows.

\vspace{-.15cm}
\begin{definition}[\bf Top-$k$ Join-Correlation Query]
Given a column~$Q$ and a join column $K_Q$ from a query table $\mathcal{T}_Q$, find the top-$k$ tables~$\mathcal{T}_X$ in a dataset collection such that $\mathcal{T}_X$ is joinable with $\mathcal{T}_Q$ on $K_Q$ and has the highest after-join correlations between a column $C \in \mathcal{T}_X$ and~$Q$.
\end{definition}
\vspace{-.15cm}

Exact top-$k$ join-correlation queries can be answered by finding all tables $\mathcal{T}_X$ that are joinable on $K_Q$, performing a full join, and finally finding the tables that have the $k$ greatest correlations. As we discussed, this is inefficient and does not scale when querying large dataset collections.
Instead, we propose an efficient approach to compute approximate answers for these queries that uses \corrsketchnormal to \emph{rank the results based on correlation estimates}.
%

The possibility of error presents a challenge for answering approximate join-correlation queries. While searching for columns correlated with query column $Q$ within a large collection of disconnected tables, there will typically be many more poorly-correlated columns than highly-correlated ones. 
In this ``needle-in-a-haystack'' setting, estimation errors will lead to many false positive results: inevitably some poorly-correlated columns will look far more correlated with $Q$ than they actually are, even more correlated than the columns we aim to find. Therefore, simply ranking results based on the correlation estimates may eventually produce poor rankings.
In the remainder of this section,
we propose a framework for scoring columns that addresses this problem.

Another important aspect of an implementation of our method is query evaluation.
While \corrsketchnormal efficiently estimate join-correlations, a na\"ive approach that computes correlations for all possible column pairs can still be expensive for large collections.
Computing correlations for all pairs, however, is not necessary: not all possible join keys have a high key overlap to yield useful joins.
To efficiently answer top-$k$ join-correlation queries we can leverage efficient data structures and query processing algorithms for set overlap search.
Recall that a sketch includes a set of pairs $\langle h(k), x_k \rangle$.
Since $h(k)$ is a discrete value, we can leverage existing data structures for efficient querying such as inverted indexes available in off-the-shelf systems (e.g., PostgreSQL, Apache Lucene) and efficient query processing algorithms for set similarity search such as JOSIE \cite{zhu@sigmod2019}, ppjoin+ \cite{xiao2011efficient}, CRSI \cite{venetis2012crsi}, and Lazo \cite{fernandez@icde2019}.

\subsection{Ranking with Uncertain Estimates}
\label{sec:ranking-with-estimates}

To take into account the uncertainty associated with the estimations, we adopt a simple 
risk-averse scoring framework that selects  $k$ results and maximizes:
%
\vspace{-.15cm}
\begin{equation} \label{eq:risk-aversion-scoring}
\max \sum_{i=1}^{k} \Big( |\hat{r}_{Q \bowtie C_i}| * \big(1 - risk(Q, C_i) \big) \Big)
\vspace{-.1cm}
\end{equation}
where $|\hat{r}_{ Q \bowtie C_i }|$ is the absolute value of the correlation estimate computed using a correlation estimator applied to $L_{Q \bowtie C_i}$, and $risk(Q, C_i)$ is a function that returns a number in the range $[0,1]$ and measures the dispersion of the correlation estimates using $L_{Q \bowtie C_i}$,  such as standard error or confidence interval length.
Intuitively, whenever the risk associated with the estimation is non-zero, a penalty factor proportional to the risk is applied to the estimate.
Note that we consider the absolute of the correlation in Eq. \ref{eq:risk-aversion-scoring}, given that negative correlations can be as useful as positive correlations.

\vspace{-.1cm}
\subsection{Measuring the Estimation Error Risk}
\label{sec:estimation-risk}


%

A natural approach to measure the risk of estimation error is to use standard statistics such as the standard error of the estimator or the length of the confidence interval.
In our scenario, however, we do not have access to these: 
since we do not explicitly join the columns (given that doing so is computationally prohibitive), we have little information about the distributions of column values.
For instance, it is not possible to compute the exact variance (and similarly standard deviation or standard error) for the correlation because it depends on the underlying correlation of the complete data and fourth-order moments of the variables~\cite{yuan@jmva2000} (see Equation~\ref{eq:pearsons-variance} for the variance under normal distribution assumptions or see \cite{bowley1928standard} for the general case). Moreover, estimating high-order moments using small sample sizes may be unstable~\cite{bishara@bjmsp2018}.

One statistic that we can compute is the standard error of the sampling distribution of the Fisher's Z correlation transformation~\cite{berry2000monte}, given by $SE_z = 1/\sqrt{n-3}$. While it assumes a bivariate normal distribution, its computation is simple and only depends on the known sample size $n$.
Despite the normality assumption, this statistic is asymptotically equivalent to the error found in our theoretical analysis of $1/\sqrt{n}$ (Section~\ref{sec:confidence-interval}). Thus, we expect it to work increasingly well as the sample size increases for any data distribution.

While we expect the Fisher's Z standard error to be accurate for large sample sizes, one drawback is that it assumes normality and does not take into account any information about the data distribution.
Given that real-world data is seldom normally distributed~\cite{micceri1989unicorn} and the actual data distributions are usually unknown, we are interested in calculating distribution-independent confidence interval bounds.
In this setting, non-parametric approaches such as bootstrapping are applicable. 
Bootstrapping only makes use of the samples generated by our sketches and does not assume any prior data distribution. 
While bootstrapping has been shown to have good performance for estimating confidence intervals for non-normal distributions~\cite{bishara@bjmsp2018}, it has the disadvantage of having a very high computational cost -- specially in settings like ours where it needs to be computed repeatedly over many columns.

To address the limitations of Fisher's Z (normality assumption) and bootstrapping-based (high computational cost) methods, in Section \ref{sec:confidence-interval}, we derive new confidence bounds using finite-sample concentration bounds for sums of independent random variables.
These bounds only depend on the maximum and minimum values in $X_{X\bowtie Y}$ and $Y_{X\bowtie Y}$. Since $X_{X\bowtie Y} \subseteq X$ and $Y_{X\bowtie Y} \subseteq Y$, the values in these columns lie strictly within the range of values in $X$ and $Y$. Thus, bounds on the range can be computed with a single pass over the columns (i.e., at the same time we construct the \corrsketchnormal).
Given that these bounds can be computed in constant time, there is essentially no additional computational overhead.

\subsection{Confidence Interval Bounds}
\label{sec:confidence-interval} 

As discussed in Section~\ref{sec:prelim_corr_estimation}, one challenge in deriving confidence intervals is that Pearson's correlation is known to be highly sensitive to individual samples. 
This prevents directly using e.g., a standard McDiarmid’s inequality to bound the accuracy of an estimate for the correlation.
Instead, we use individual Hoeffding inequalities to obtain confidence intervals for each individual component of the correlation estimator, and then apply a union bound to combine these results into an overall confidence interval.
This is similar to the approach used for constructing confidence intervals for the sample variance for non-Gaussian data in~\cite{bardenet2015concentration}.


\paragraph{Analysis} 
\corrsketch estimates the Pearson's correlation between two columns $X_{X\bowtie Y}$ and $Y_{X\bowtie Y}$. Let $C_{low} = \min\{x \in X, y\in Y\}$ and $C_{high} = \max\{x \in X, y\in Y\}$ be upper and lower bounds in columns $X$ and $Y$, and let $C = C_{high} - C_{low}$. Let $A = X_{X\bowtie Y} - C_{low}$ and $B = Y_{X\bowtie Y} - C_{low}$. 
Since a constant shift does not affect Pearson's correlation, we note that the correlation between $X_{X\bowtie Y}$ and $Y_{X\bowtie Y}$ is equal to:
\vspace{-.1cm}
\begin{small}
$$
\rho = \frac{\langle A - \mu_A\vec{1}, B - \mu_B\vec{1} \rangle}{\|A - \mu_A\vec{1}\|_2 \|B - \mu_B\vec{1}\|_2 },
$$
\end{small}
\vspace{-.1cm}
\noindent where $\mu_A = \frac{1}{N}\sum_{i=1}^N A_i$ and $\mu_B = \frac{1}{N}\sum_{i=1}^N B_i$ are the  means of $A$ and $B$ and $\vec{1}$ is the all ones vector.
This expression is equal to
\begin{small}
$$
{\rho} = \frac{\nu_{A,B} - \mu_A\mu_B}{\sqrt{\nu_A - \mu_A^2}\sqrt{\nu_B - \mu_B^2}},
$$
\end{small}
\vspace{-.1cm}
\noindent where
$ \nu_A = \frac{1}{N}\sum_{i=1}^N A_i^2 $,
$ \nu_B =\frac{1}{N}\sum_{i=1}^N B_i^2 $
and
$ \nu_{A,B} = \frac{1}{N}\langle A,B\rangle = \frac{1}{N}\sum_{i=1}^N A_i B_i $.
%
Let $a,b\in \R^n$ be vectors containing $n$ samples drawn uniformly without replacement from $A,B$. According to Theorem~\ref{thm:uniform}, the estimate to Pearson's correlation obtained via our sketches is equivalent to:
\vspace{-.2cm}
$$
r = \frac{\nu_{a,b} - \mu_a\mu_b}{\sqrt{\nu_a - \mu_a^2}\sqrt{\nu_b - \mu_b^2}},
$$
where $\mu_a = \frac{1}{n}\sum_{i=1}^n a_i$, $\mu_b = \frac{1}{n}\sum_{i=1}^n b_i$,  $\nu_a = \frac{1}{n}\sum_{i=1}^n a_i^2$, $\nu_b = \frac{1}{n}\sum_{i=1}^n b_i^2$, and $\nu_{a,b} = \frac{1}{n}\sum_{i=1}^n a_i b_i$.

\paragraph{Union Bound}
We want to compute a confidence interval for $\rho$, meaning that, for some specified $\alpha$ (e.g., $\alpha = .05$), our goal is to compute upper and lower bounds $\rho^{low}, \rho^{high}$ depending on our estimate $r$ such that:
$
\Pr[\rho^{low} \leq \rho \leq \rho^{high}] \geq (1-\alpha).
$
To do so, we first compute the upper and lower bounds for $\mu_A$, $\mu_B$, $\nu_A$, $\nu_B$, $\nu_{A,B}$,
for all of the 5 parameters $r$ depends on. For any parameter $c$, we want: 
$
\Pr[c^{low} \leq c \leq c^{high}] \geq (1-\alpha/5).
$
For example, that $\Pr[\mu_A^{low} \leq \mu_A \leq \mu_A^{high}] \geq (1-\alpha/5)$.
We will discuss how to obtain these bounds shortly, but for now we show how they can be used to compute a confidence interval. First let:
\begin{align*}
    num_{low} &= \nu_{A,B}^{low} - \mu_{A}^{high}\mu_{B}^{high} & num_{high} &= \nu_{A,B}^{high} - \mu_{A}^{low}\mu_{B}^{low}
\end{align*}
\small
$$den_{low} = \sqrt{\max\left[0,\nu_A^{low} - (\mu_A^{high})^2\right]\cdot\max\left[0,\nu_B^{low} - (\mu_B^{high})^2\right]}$$
$$ den_{high} = \sqrt{\max\left[0,\nu_A^{high} - (\mu_A^{low})^2\right]\cdot\max\left[0,\nu_B^{high} - (\mu_B^{low})^2\right]} $$
\normalsize
Then set:
\vspace{-.2cm}
\begin{align}
\label{eq:conf_loew}
\rho^{low} = \begin{cases}
\frac{num_{low}}{den_{high}} &\text{if } num_{low}\geq 0\\
\frac{num_{low}}{den_{low}} &\text{if } num_{low}< 0
\end{cases}
\end{align}
\begin{align}
\label{eq:conf_upper}
\rho^{high} = \begin{cases}
\frac{num_{high}}{den_{low}} &\text{if } num_{high}\geq 0\\
\frac{num_{high}}{den_{high}} &\text{if } num_{high}< 0.
\end{cases}
\end{align}
By a union bound, we have $\Pr[\rho^{low} \leq {\rho} \leq \rho^{high}] \geq (1-\alpha)$.

\paragraph{Individual Parameter Bounds}
The next step is to obtain the confidence intervals for each of the parameters $\{\mu_A, \mu_B,\nu_A,\nu_B,\nu_{A,B}\}$, which we do using Hoeffding's concentration inequality for bounded random variables. This bound is usually stated for sampling \textit{with} replacement (i.e., independent random sampling), but Hoeffding proves in his original paper that it also holds for \textit{without} replacement sampling, which only gives better concentration. Specifically Theorems 2 and 4 in \cite{hoeffding} give:

\begin{lemma}[Hoeffding's Inequality]
Let $X_1, \ldots, X_N$ be a set numbers bounded $\in [0,C]$ with mean $\mu_{X} = \frac{1}{N}\sum_{i=1}^N X_i$. Let $Y_1, \ldots, Y_n$ be drawn independently without replacement from this set. Then:
$$
\Pr \left[ \left|\mu_X - \frac{1}{n}\sum_{i=1}^n Y_i\right| \geq t \right] \leq 2e^{-2nt^2/C^2}.
$$
\end{lemma}
Since $A$ and $B$ have values in $[0,C]$ by definition, each of the terms $\{\mu_A, \mu_B,\nu_A,\nu_B,\nu_{A,B}\}$ is the average of $N$ numbers, bounded between $[0,C]$ for $\mu_A, \mu_B$ and between $[0,C^2]$ for the others. 
Accordingly, we can apply Hoeffding's inequality to obtain a confidence interval with with $\alpha/5$, as required by our analysis above. For example, consider $\mu_A$. We have from Hoeffding's that:
$$
\Pr[\left|\mu_A - \mu_a\right| \geq t] \leq 2e^{-2nt^2/C^2}.
$$
For $2e^{-2nt^2/C^2} = \alpha/5$, we solve for $t = \sqrt{\ln(10/\alpha)\cdot C^2/2n}$ and can then set $\mu_A^{low} = \mu_a - t$ and $\mu_A^{high} = \mu_a + t$. As another example, consider $\nu_A$. From Hoeffding's,
$$
\Pr[\left|\nu_A - \nu_a\right| \geq t] \leq 2e^{-2nt^2/C^4}.
$$
For $2e^{-2nt^2/C^4} = \alpha/5$, we solve for $t' = \sqrt{\ln(10/\alpha)\cdot C^4/2n}$ and can then set $\nu_A^{low} = \nu_a - t'$ and $\nu_A^{high} = \nu_a + t'$.  
Final bounds for all five parameters are as follows:
\small
\begin{align*}
[\mu_A^{low},\mu_A^{high}] &= [\mu_a - t,\mu_a + t], & [\mu_B^{low},\mu_B^{high}] &= [\mu_b - t,\mu_b + t] 
\end{align*}
\begin{align*}
[\nu_A^{low},\nu_A^{high}] &= [\nu_a - t',\nu_a + t'], & [\nu_B^{low},\nu_B^{high}] &= [\nu_b - t',\nu_b + t'] 
\end{align*}
\begin{align*}
[\nu_{A,B}^{low},\nu_{A,B}^{high}] &= [\nu_{a,b} - t',\nu_{a,b} + t']
\end{align*}
\normalsize
where $t = \sqrt{\ln(10/\alpha)\cdot C^2/2n}$ and $t' = \sqrt{\ln(10/\alpha)\cdot C^4/2n}$.

\paragraph{Discussion}
The above analysis gives a simple procedure to compute valid confidence interval for correlation estimates: we first compute $t,t'$ as defined above, which only requires the desired confidence level $\alpha$, $C$ which is pre-computed, and the sample size $n$ (i.e., the number of rows in $L_{\langle X\bowtie Y \rangle}$. Then, we compute upper and lower bounds for each parameter, and plug into Equations~\eqref{eq:conf_loew} and \eqref{eq:conf_upper} to obtain a final confidence interval. 

Our analysis shows that, for a fixed $1 - \alpha$ confidence level, the confidence interval bounds depend on the sample size and up to the fourth-power of the range $C$ of the variables $A$ and $B$. This is in line with previous findings which point that the deviation of the Pearson's sample estimator depends on the fourth-moment of the variables \cite{yuan@jmva2000, yuan2005effect}. The accuracy of the bound also scales inversely with the square root of the sample size $n$, as expected.

Correlation is difficult to estimate whenever the population variance is low, because the estimator can become unstable (as these terms appear in the denominator of the equation for $r$). In particular, if either $\|A - \mu_A\vec{1}\|_2$ or $\|B - \mu_B\vec{1}\|_2$ is close to zero, then  any small deviation of $\|a - \mu_a\vec{1}\|_2$  or $\|b - \mu_b\vec{1}\|_2$ from the true values that we are trying to estimate would lead to a large difference in $\rho$ vs. $r$.

So, to get a better sense of the bounds, let us assume that, for both $A$ and $B$, $var(A) = \nu_A - \mu_A^2$ and $var(B) = \nu_B - \mu_B^2$ are $\geq  c$ for some constant $c$. From the analysis above, it is relatively easy to show that if we set $n = O\left(\frac{C^4\ln(1/\alpha)}{\epsilon^2 c^2}\right)$, we will obtain a final confidence interval with width $2\epsilon$ -- i.e., we can estimate the Pearson's correlation to accuracy $\pm \epsilon$. For data bounded by $C$, it would be natural for the variance $c$ to be on the order of $C^2$, in which case, the number of samples required is just $n = O\left(\frac{\ln(1/\alpha)}{\epsilon^2}\right)$. In other words, as the sample size $n$ grows (i.e., as our sketches size increases) our
estimate converges with error roughly
$\frac{1}{\sqrt{n}}$. 

\paragraph{Effect of Small Sample Sizes}
Note that when sample sizes are small, the bounds for the standard deviation terms ($\nu-\mu^2$) may become negative, which makes $den_{high}$ and $den_{low}$ to become zero and thus yield invalid bounds. To address this problem, while computing $\rho^{low}$ and $\rho^{high}$, we can replace the denominator of the Equations~\ref{eq:conf_loew} and \ref{eq:conf_upper} by the product of the sample standard deviation of the variables computed using the samples induced by the sketch join, i.e., we can set $den_{high} = den_{low} = \sqrt{\nu_a-\mu_a^2} \sqrt{\nu_b-\mu_b^2}$. 
We refer to these modified upper bound and lower bound, respectively, as $\rho^{high}_{HFD}$ and $\rho^{low}_{HFD}$.
Although these are not true probabilistic bounds, their resulting confidence interval length ($\rho^{high}_{HFD} - \rho^{low}_{HFD}$) still provides meaningful information to measure the estimation error risk. This is because the denominator term serves as a normalization factor of the covariance term in the numerator (for which we still compute the true bounds).
\subsection{Scoring Functions} 
\label{sec:scoring-functions}
%
%
Based on the framework (Section \ref{sec:ranking-with-estimates}) and statistics (Sections \ref{sec:estimation-risk} and \ref{sec:confidence-interval}) described above, we finally derive four different scoring functions that optimize ranking for correlated column discovery.
Let $ci_{length} = \rho^{high}_{HFD} - \rho^{low}_{HFD}$ be the confidence interval length of a correlation estimate. Then, set $ci^{high}_{max}$ and $ci^{high}_{min}$ to be the maximum and the minimum confidence interval length in a ranked list, respectively. We define the following risk penalization factors:
\begin{small}
\vspace{-1em}
\begin{align*}
se_{z} = 1 - \frac{1}{\sqrt{\max(4, n) - 3}} && ci_{b} = 1 - \frac{\rho^{high}_{PM1} -\rho^{low}_{PM1}}{2}
\end{align*}
\vspace{-1em}
\begin{align*}
ci_{h} = 1 - \frac{ci_{length} - ci^{high}_{min} }{ ci^{high}_{max} - ci^{high}_{min} }
\end{align*}
\end{small}
Recall from our framework that each of these factors assume values in the range $[0,1]$ and when they are equal to $1$, the error risk is the minimum possible.
The equations above are direct applications of the statistics introduced in Sections \ref{sec:estimation-risk}-\ref{sec:confidence-interval}, i.e., Fisher's Z transformation standard error, Bootstrap CI, and our Hoeffding's~CI.
By plugging them into Equation~\ref{eq:risk-aversion-scoring}, we have (in addition, we also consider a no-penalization factor in $s_1$):
\begin{align*}
& s_1 = r_p,           &   s_2 = r_p*se_{z}, \\ 
& s_3 = r_b*ci_{b},    &   s_4 = r_p*ci_{h},
\end{align*}
where $r_p$ and $r_b$ are the absolute of the correlation estimated using, respectively, the Pearson's sample estimator and the PM1 bootstrap estimator~\cite{wilcox1996confidence}. We also use PM1 for the confidence intervals in $ci_b$.

\section{Experimental Evaluation}
\label{sec:experiments}

We performed extensive experiments using both synthetic and real-world \datasets to evaluate the effectiveness of \corrsketch and the proposed ranking strategies. 

\subsection{Datasets}

We used three different data collections:
\textit{Synthetic Bivariate Normal}, which is  synthetically generated and follows  a pre-defined, well-known data distribution;
the \textit{World Bank's Finance}~\cite{wbopendata-finance} and the \textit{NYC Open Data}~\cite{nycopendata} collections, which contain real-world datasets, currently published in open data portals. 
We used snapshots of these repositories collected in September 2019 using Socrata's REST API~\cite{socrata}. All datasets were stored in plain CSV text files, and we used the \textit{Tablesaw}
library~\cite{tablesaw} to automatically parse and detect the basic data types for each column. 

\paragraph{\wbfdataset (\wbfdatasetshort)}
This collection contains 64 datasets (tables) related to the World Bank's Finances~\cite{wbopendata-finance}.
There is missing data in several columns and some columns contain large monetary values. From each table, we extracted all possible pairs of categorical and numerical data columns $\langle K_X, X\rangle$, out of which we generated all possible unique 2-combinations of columns pairs --  9,979,278 pairs of column pairs, i.e., $\langle K_X, X\rangle$, $\langle K_Y, Y\rangle$.

\paragraph{\nycdataset (\nycdatasetshort)}
The tables from this \dataset contain data published by New York City agencies and their partners~\cite{nycopendata}.
Our snapshot includes 1,505 different \datasets.
Using the same process described above for the \wbfdataset collection, we generated $12,497,500$ pairs of column key-value pairs.

\begin{figure}[t]

\parbox{.45\linewidth}{
\begin{subfigure}{1\linewidth}
    \centering
    \includegraphics[width=1\linewidth]{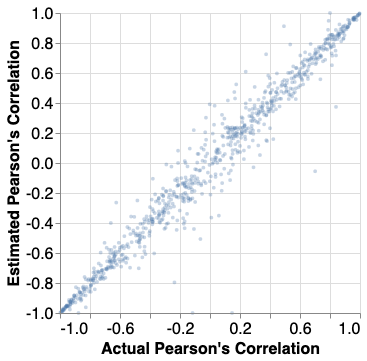}
    \caption{SBN dataset, $n\geq3$.}
\end{subfigure}
}
\hspace{0.05\linewidth}
\parbox{.45\linewidth}{
\begin{subfigure}{1\linewidth}
    \centering
    \includegraphics[width=1\linewidth]{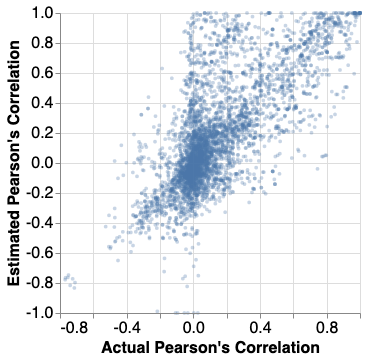}
    \caption{WBF dataset, $n\geq3$.}
\end{subfigure}
}


\parbox{.45\linewidth}{
\begin{subfigure}{1\linewidth}
    \centering
    \includegraphics[width=1\linewidth]{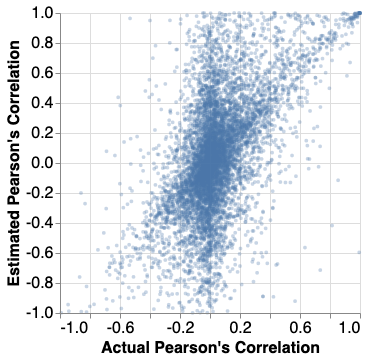}
    \caption{NYC dataset, $n\geq3$}
\end{subfigure}
}
\hspace{0.05\linewidth}
\parbox{.45\linewidth}{
\begin{subfigure}{1\linewidth}
    \centering
    \includegraphics[width=1\linewidth]{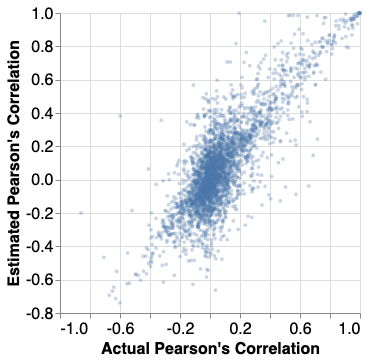}
    \caption{NYC dataset, $n\geq20$}
\end{subfigure}
}
\caption{Estimation errors significantly vary for different samples sizes and different datasets with different data distributions. (a), (b), and (c) show the deviations of all column pairs for 3 different datasets. (d) shows the estimates from (c) after filtering out estimates that use fewer than 20 samples.}
\vspace{-.6cm}
\label{fig:scatterplots-all-datasets}
\end{figure}

\paragraph{\sbndataset (\sbndatasetshort)} This dataset was generated by creating $t$ tables consisting of $n$ tuples $\langle k, x_k, y_k \rangle$, where $k \in K$ is a random unique string, and $x_k \in X$ and $y_k \in Y$ are real numbers drawn from a bivariate normal distribution with mean $\mu=0$. The covariance of the column vectors $X$ and $Y$ was chosen in such a way that the Pearson's correlation coefficient between $X$ and $Y$ was approximately equal to a given parameter $r_{XY}$. We then created $t$ pairs of tables $\mathcal{T}_X = \langle K_X, X \rangle$ and $\mathcal{T}_Y = \langle K_Y, Y \rangle$. Finally, we reduced the size of the table $\mathcal{T}_Y$ from $n$ to $n'$ by selecting a uniform random sample of size $n'= n*c$, where $c$ is a random real number in the range $(0, 1)$ indicating the join probability between $X$ and $Y$. We set the number of table pairs 
$t=3000$. For each table pair, we set $n$ to be a random number drawn uniformly in the range~$(0, 500000)$, and the correlation $r_{XY}$ is drawn uniformly at random from~$(-1, 1)$.

\subsection{Correlation Estimation Accuracy}


To study the effectiveness of our  sketching method, we compare the estimated correlations against the actual correlations. 
For each pair of columns, $\langle K_X, X\rangle$ and $\langle K_Y, Y\rangle$, we first build their  correlation sketches $L_{\langle K_X, X \rangle}$ and $L_{\langle K_Y, Y \rangle}$, and then compute their correlation estimates $\hat{r}_{\langle X \bowtie Y \rangle}$. Then, we compare the estimates with the actual column correlations $r_{X \bowtie Y}$ computed using the (complete) join of columns $X_{X \bowtie Y}$  and $Y_{X \bowtie Y}$. Due to space limitations, we show the results for Pearson's correlation. We discuss the performance of other estimators in Section~\ref{sec:exp-estimators}.

Figure~\ref{fig:scatterplots-all-datasets} shows scatterplots for the estimates computed for all three \datasets against the actual Pearson's correlation values, using sketches of size 256. We can clearly see that \corrsketch produces quite accurate results for the \sbndatasetshort \dataset (Figure~\ref{fig:scatterplots-all-datasets}a), which contains only data drawn from a bivariate normal distribution: the estimates are concentrated close to their actual correlation values with only a few points deviating from the actual correlation values.

For the \nycdatasetshort (Figure~\ref{fig:scatterplots-all-datasets}b) and \wbfdatasetshort (Figure~\ref{fig:scatterplots-all-datasets}c) datasets, which contain real-world data from unknown distributions, as expected, there are more incorrect predictions. 
We can observe that for many points with actual correlation equal to $0$, the correlation is  overestimated -- see the vertical line around the value zero in the x-axis. 
This happens because some of the joins computed from the sketches can be too small, leading to this inaccuracy. Note that Figures~\ref{fig:scatterplots-all-datasets}a, b, and c reflect estimates computed with joins (sketch intersection) that contain as few as 3 tuples ($n\geq3$).
Nonetheless, we can still observe a large concentration of points around the diagonal line, suggesting that \corrsketch is effective.

To better understand the behavior of the sketches with respect to sample size, we also plot in Figure~\ref{fig:scatterplots-all-datasets}d the
results for the \nycdatasetshort data showing only join samples with size at least $20$.
%
The plot shows that indeed, for larger sample sizes, the behavior is more similar to that of the SBN dataset -- the points are more concentrated closer to the diagonal, indicating the estimates are more accurate.

This provides evidence for the issue we discussed in Section~\ref{sec:ranking}: in large table collections, there are often many more poorly-correlated tables than well-correlated ones, and estimation errors may lead to a potentially large number of false positives.
This underscores the importance of having effective ranking functions
to help users focus their attention on tables that are more likely to be correlated. In this case, for example, a ranking function that prioritizes estimates computed from larger join samples is likely to be effective at pruning
these false positive results.

\begin{figure}[t]
\includegraphics[width=.9\linewidth]{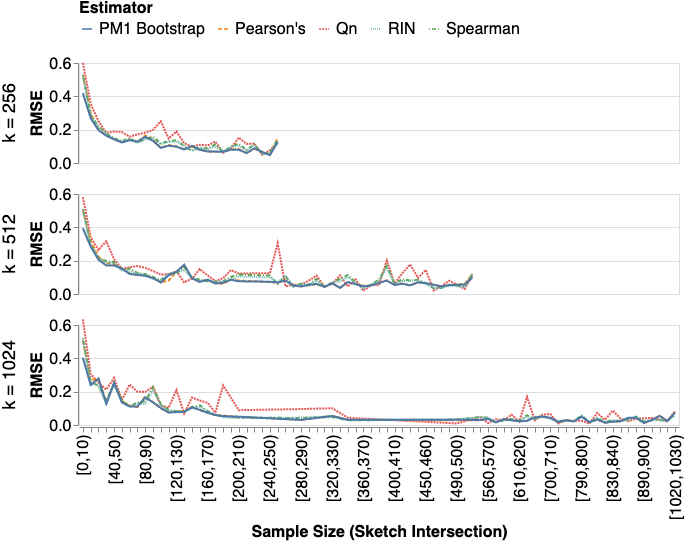}
\vspace{-.3cm}
\caption{The sample size (sketch intersection) has an impact on RMSE. As the sketch intersection size increases, the RMSE decreases in the \nycdatasetshort dataset. Here, the $k$ parameter (row) denotes the maximum sketch size (number of minimum values kept in the sketch).}
\vspace{-.6cm}
\label{fig:rmse-vs-sketch-size}
\end{figure}

\begin{table*}[t]

\begin{small}
\parbox{.245\linewidth}{
\begin{subfigure}{1\linewidth}
\centering
    \begin{tabular}{lrl}
        \toprule
               ranker &  score &       \% \\
        \midrule
         $r_p*ci_{h}$ &  0.529 &  193.2\% \\
         $r_b*ci_{b}$ &  0.516 &  185.9\% \\
         $r_p$ &  0.507 &  180.9\% \\
         $r_p*se_{z}$ &  0.420 &  133.1\% \\
         $jc$ &  0.180 &  0.0\% \\
         $\hat{jc}$ &  0.172 &  -4.8\% \\
         random &  0.161 &  -10.8\% \\
        \bottomrule
    \end{tabular}
    \caption{MAP ($r > .75$)}
\end{subfigure}
}
\parbox{.245\linewidth}{
\begin{subfigure}{1\linewidth}
\centering
    \begin{tabular}{lrl}
    \toprule
           ranker &  score &       \% \\
    \midrule
     $r_p*se_{z}$ &  0.472 &  102.1\% \\
     $r_p*ci_{h}$ &  0.467 &  99.8\% \\
     $r_p$ &  0.452 &  93.2\% \\
     $r_b*ci_{b}$ &  0.428 &  83.2\% \\
     $\hat{jc}$ &  0.239 &  2.3\% \\
     $jc$ &  0.234 &  0.0\% \\
     random &  0.202 &  -13.7\% \\
    \bottomrule
    \end{tabular}
    \caption{MAP ($r > .50$)}
\end{subfigure}
}
\parbox{.245\linewidth}{
\begin{subfigure}{1\linewidth}
\centering
    \begin{tabular}{lrl}
    \toprule
           ranker &  score &      \% \\
    \midrule
     $r_b*ci_{b}$ &  0.714 &  51.5\% \\
     $r_p*ci_{h}$ &  0.705 &  49.5\% \\
     $r_p$ &  0.699 &  48.4\% \\
     $r_p*se_{z}$ &  0.689 &  46.2\% \\
     random &  0.481 &  2.1\% \\
     $\hat{jc}$ &  0.480 &  1.8\% \\
     $jc$ &  0.471 &  0.0\% \\
    \bottomrule
    \end{tabular}
    \caption{nDCG@5}
\end{subfigure}
}
\parbox{.245\linewidth}{
\begin{subfigure}{1\linewidth}
\centering
    \begin{tabular}{lrl}
    \toprule
           ranker &  score &      \% \\
    \midrule
     $r_b*ci_{b}$ &  0.845 &  17.7\% \\
     $r_p$ &  0.843 &  17.5\% \\
     $r_p*ci_{h}$ &  0.841 &  17.3\% \\
     $r_p*se_{z}$ &  0.832 &  15.9\% \\
     $\hat{jc}$ &  0.726 &  1.2\% \\
     random &  0.724 &  0.9\% \\
     $jc$ &  0.717 &  0.0\% \\
    \bottomrule
    \end{tabular}
    \caption{nDCG@10}
\end{subfigure}
}
\end{small}

\caption{Ranking evaluation scores in terms of MAP and nDCG. The ``\%'' column denotes relative improvement over $jc$.}
\vspace{-.6cm}
\label{tab:ranking-metrics}
\end{table*}

\vspace{-.2cm}
\subsection{Exploring Different Correlation Estimators} \label{sec:exp-estimators}
We studied the performance of different correlation estimators:

\paragemph{(1) Pearson's Sample Correlation:} Computes the correlation using the formula defined in Equation \ref{eq:pearson-sample}.

\paragemph{(2) Spearman's Rank Correlation:} Let $r(x)= r_x$ where $r_x = 1$ for the smallest $x$, $r_x = 2$ for the second smallest $x$, and so on. The numeric column values are  transformed using $r(x)$ and then the Pearson's correlation over the transformed values is computed.

\paragemph{(3) Rank-based Inverse Normal (RIN):} Similarly to Spearman's, a transformation function is applied before computing the Pearson's correlation. Following  \cite{bishara-hittner@epm2015}, we employ the \textit{rankit} function \cite{bliss1967statistics} 
which is defined as: $h(x) = \Phi^{-1} \left(\frac{r(x)-1/2}{n} \right),$ where $\Phi$ is the inverse normal cumulative distribution function.

\paragemph{(4) $Q_n$ correlation:} Computes the correlation using as modified formula of Pearson's correlation and the robust $Q_n$ scale estimator. For more details, see \cite{shevlyakov@robust-corr-book}.

\paragemph{(5) PM1 Bootstrap:} Performs repeated re-sampling with replacement of the data and recomputes the correlation using the Pearson's sample correlation estimator. The average of all re-samples is then used as an estimate. Instead of drawing a fixed large number of re-samples (say 10.000), we stop the re-sampling when the probability of changing the mean by more than $0.01$ falls below $0.05\%$.

Most computed correlation estimates using sketches were compared to their corresponding population correlations (i.e., including the transformations of the population data when applicable). The only exception was the \textit{PM1 Bootstrap} estimator, which
is compared to the population's Pearson's correlation that it intends to estimate.

%
Correlation estimators differ in their sample size requirements and sensitivity to the data distribution~\cite{bonett@psychometrika2000}. To better understand how the correlation estimation accuracy changes when we vary the estimator and the amount of storage space used by the sketch (determined by the maximum sketch size), we plot the root mean squared error (RMSE) for different choices of correlation estimators and maximum sketch sizes in Figure~\ref{fig:rmse-vs-sketch-size}. We can see a trend: for all estimators and maximum sketch sizes, the RMSE decreases as the intersection size between the sketches increases. The RMSE stabilizes roughly at $0.1$. While the different estimators display similar trends, the plot also shows that some estimators are less robust (see e.g., the spikes in the line for $Q_n$).

\vspace{-.2cm}
\subsection{Correlated Column Ranking}
To better understand the effectiveness of the proposed scoring functions, we use the NYC data collection which contains the largest number of tables. For each pair of columns $\langle K_X, X \rangle$ in the collection, we retrieved all other  joinable columns $\langle K_Y, Y \rangle$. Then, we ranked the list of retrieved columns using the scoring functions described in Section~\ref{sec:ranking}. As baselines we use: 1) a \textit{random} scoring function, which assigns random scores in the range $[0,1]$ drawn from a uniform distribution; 2) the \textit{exact} Jaccard Containment ($jc$) similarity computed using the complete data after the join; and 3) the JC similarity estimated ($\hat{jc}$) using our correlation sketches.
Even though these baselines do not take correlation into account, we use them as a point of comparison, since they represent existing methods for retrieving joinable tables.

We use two well-known measures to compare the scoring functions: mean average precision (MAP) and normalized discounted cumulative gain (nDCG)~\cite{jarvelin2002cumulated}. nDCG supports graded relevance judgments, i.e., each retrieved column can receive different scores depending on their relevance (the absolute value of the correlation, in our case) and position on the list. In contrast, MAP supports binary judgment, and relies on thresholds to decide which columns are relevant (say, $r>.5$).

To evaluate different aspects of the ranking, we compute MAP for the whole ranked list using different thresholds for correlation relevance: $r>0.50$ and  $r>0.75$. For nDCG, which has a tendency of assigning higher scores, we compute the metrics only over the top-$k$ search results, where $k=5$ and $k=10$. 
Table~\ref{tab:ranking-metrics} summarizes the results. The first trend we observe is that the ranking produced using the JC similarity scores attains scores similar to random ordering. This  confirms our hypothesis that JC is not well suited for ranking the results of join-correlation queries and also suggests that highly-correlated columns can have different levels of JC similarity.

Note that all correlation-based ranking functions (Section~\ref{sec:scoring-functions}) show  significant improvements over the baselines. The ranking functions based on our Hoeffding bounds ($r_p*ci_h$) attain scores that are either better or very close to the bootstrap-based ranking function. This is specially interesting due to the fact that the Hoeffding-based CI 
can be computed in constant time, in contrast to the bootstrap method that requires many re-samples and the computation of correlations for these samples (typically, around between 1,000 and 10,000 iterations are used~\cite{bishara@bjmsp2018, bishara@brm2017}). In other words, we derive rankings that are comparable to those produced by the bootstrapping at a fraction of the cost.
Moreover, the MAP scores suggest that the Hoeffding-based scoring is particularly effective at avoiding false positives with correlation above $r>0.75$.

\begin{figure}[t]
\includegraphics[width=1\linewidth]{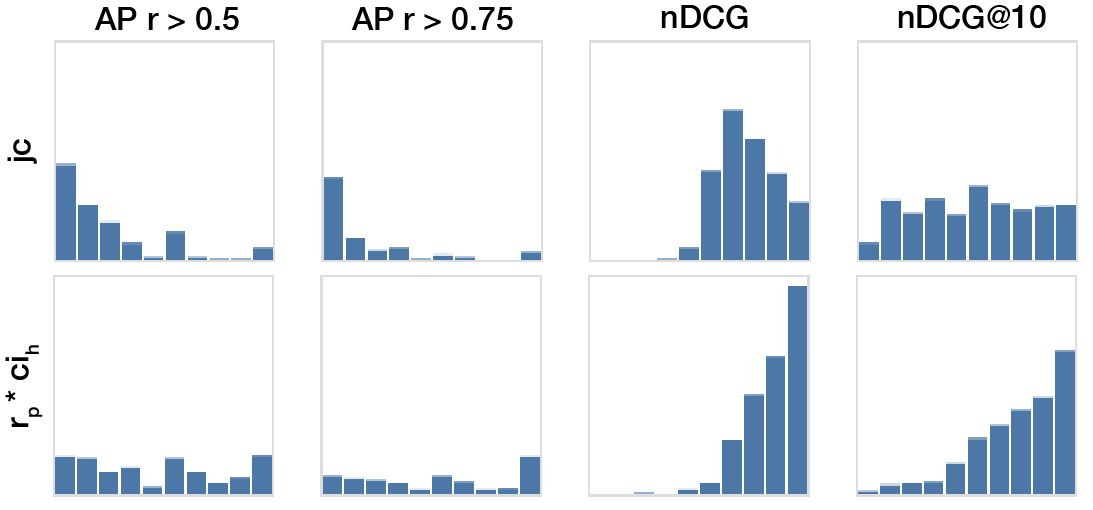}
\vspace{-0.3cm}
\caption{Distribution of the evaluation metric scores for different scoring functions. $x$-axis shows slices of the metric range $[0,1]$. Each bar corresponds to a slice of width $0.1$. The $y$-axis show the number of queries that fall in each slice.}
\vspace{-0.6cm}
\label{fig:ranking-scores-distributions}
\end{figure}

Note that Table~\ref{tab:ranking-metrics} shows only average scores over all queries. To get a better sense for the improvements over all queries, we plot the distribution of evaluation scores in Figure~\ref{fig:ranking-scores-distributions}.
The histograms show the number of queries for different metric score values. Rows display the scoring functions: JC similarity (first row) and the Hoeffding-based scoring (second row). We can clearly see that, for all metrics, the scores distributions shift from left (bad scores) to  right (good scores), indicating improvements in the whole metric range. It is also clear that, specially when considering the nDCG metric, most queries have very good scores (i.e., close to the optimal).

\subsection{Running Time}

\paragraph{Join-Correlation Estimation}
To assess the efficiency of our approach and suitability for use
in dataset search engines, we compare the running times of joins and correlation estimation computations using the sketches
against the times required to execute the 
same computation over the full data. 
Since \corrsketch creates
fixed-sized synopses for arbitrarily-sized datasets, it is natural to expect a big performance impact due to the reduction in the complexity of the problem -- from potentially large table sizes to a small constant factor (sketch size).
The results are summarized in Table~\ref{table:running-time}. Due to space limitations we only include the times for joins and correlation computations using Pearson's and Spearman's coefficients. The results confirm that, with sketches, queries can be evaluated orders of magnitude faster than by using the full data, and given that the sketch size is fixed, their running time is more predictable.
Note that, for this comparison, we assume the data is loaded in memory. However, for large datasets, the costs associated with using the full data would be even higher due to the cost of reading data from disk or transferring it over a network.

\paragraph{Query Evaluation}
We also assessed the performance of evaluating join-correlation queries. 
For this experiment, we extracted all column pairs $\langle K_X, X \rangle$ from all 1,505 tables in the NYC dataset and randomly split them into two distinct sets, which we denote as \textit{query set} and \textit{corpus set}.
Next, we set the maximum sketch size to $1024$ and
built an inverted index for all tables of the corpus set using
a standard indexing library~\cite{apache-lucene}.
%
Finally, we used all pairs from the \textit{query set} to issue queries against the index and measured the query execution time. We observed that the running times for 94\% of queries are under 100 ms and \texttildelow 98.5\% of queries take under 200 ms. These times include retrieving the top-100 columns by key overlap, reading sketches from the index, and re-sorting all columns by estimated correlation. 
These preliminary results are promising and suggest that our approach can provide interactive response times for join-correlation queries over large data collections.

\begin{table}[t]
\centering
\begin{small}
\begin{tabular}{lrrrrrr}
\toprule
\textbf{} & \multicolumn{3}{c}{\textbf{Full data}} & \multicolumn{3}{c}{\textbf{Sketch}} \\
{percentiles} &  join &  $r_s$ &  $r_p$ &  join &    $r_p$ &    $r_s$ \\
\midrule
mean  & 42.219 & 8.494 & 0.240 & 0.026 & 0.000 & 0.004 \\
std. dev.   & 367.696 & 134.357 & 9.314 & 5.618 & 0.042 & 0.279 \\
75\%   & 0.231 & 0.141 & 0.005 & 0.003 & 0.000 & 0.002 \\
90\%   & 7.038 & 0.154 & 0.011 & 0.006 & 0.001 & 0.004 \\
99\%   & 1360.605 & 29.583 & 0.385 & 0.012 & 0.003 & 0.013 \\
99.9\% & 4021.838 & 2731.154 & 51.278 & 0.021 & 0.007 & 0.033 \\
\bottomrule
\end{tabular}
\end{small}
\vspace{0.1cm}
\caption{Running times (in milliseconds) for computing joins and correlations using full data and sketches. Using the full data, queries can take orders of magnitude more time than when using the sketches. $r_s$ denotes the Spearman's estimator and $r_p$ denotes Pearson's estimator.}
\label{table:running-time}
\vspace{-.9cm}
\end{table}

\section{Related Work}
\label{sec:related}

Our work builds on
a large body of research that includes sampling and sketching algorithms used to handle massive data streams~\cite{cormod@foundtrends2012}, join size estimation~\cite{harmouch@vldb2017cardest-survey}, 
and the statistics literature on
correlations.
Also related to our work is  recent research on dataset search~\cite{noy@www2019, zhang@sigir2018}, and methods that support data augmentation queries which, given a query dataset, find datasets that can be joined or concatenated \cite{zhu@sigmod2019, nargesian@vldb2018, zhu@vldb2016-lsh-ensemble, fernandez@icde2019, yang@icde2019} or 
that contain related entities~\cite{lehmberg@jws2014, zhang2020web}. 


\paragraph{Cardinality Estimation}
Several methods have been proposed that summarize massive datasets into succinct data structures to support approximate queries using bounded memory~\cite{cormod@foundtrends2012}. 
A widely-studied problem in this area is cardinality estimation, i.e., estimating the number of distinct elements in a set.
Approaches to this problem can be broadly classified into two groups: 
sampling and sketch-based algorithms. 
\textit{Sampling algorithms} avoid scanning the full dataset and estimate cardinalities based  on data samples. These algorithms have been criticized for their inability to provide a good error guarantees~\cite{charikar@pods2000error-guarantees,haas@vldb95}.
In contrast, \textit{sketch-based algorithms} fully scan the dataset once, compute hashes of the items, and create a sketch that can be used to compute cardinality estimates~\cite{bar-yossef@random2002,beyer@sigmod2007, duffield@jacm2007, cohen@vldb2008, beyer@cacm2009, dasgupta@icdt2016,  ting2016towards, cohen@kddd2017-minimal-variance}. 
Harmouch and Naumann~\cite{harmouch@vldb2017cardest-survey} categorized algorithms these into different  families: \textit{count trailing 1s}, \textit{count leading 0s}, \textit{$k^{th}$ minimum values}, and \textit{linear synopses}. Each of these families has their best algorithms which exhibit different trade-offs, which were studied in detail in~\cite{harmouch@vldb2017cardest-survey}.

While the best algorithms based on counting trailing 1s an 0s (such as HyperLogLog (HLL) \cite{flajolet2007hyperloglog})
are able to provide better accuracy per bit, 
it is not clear how they can be extended to estimate cardinalities of value-sets that satisfy arbitrary properties not known apriori, before the sketch is built. %
%
Algorithms from the $k^{th}$ minimum values (KMV) family, on the other hand, can be extended for this purpose
at the cost of storing the sample identifiers in addition to their hashed values~\cite{dasgupta@icdt2016-theta-sketches}.  The reason why HLL-like algorithms cannot estimate such properties is that, HLL does not maintain any sample of identifiers from the data. For this same reason, HLL sketches are not
suitable for join-correlation sketches, which require alignment of numeric values based on their join key values.
\corrsketch builds on the KMV family and derives a new sketch that is able to align numeric values based on their key hashed values to reconstruct a random samples that can ultimately be used for estimating correlations.
While we focus on correlations in this paper, the sketches we propose can be used to compute any statistics that are based on paired numeric values (e.g., distance correlations~\cite{szekely@tas2007} and the entropy-based mutual information).

\paragraph{Join Size Estimation}
%
The ability to accurately estimate join sizes is crucial for query optimizers.
Acharya~et~al. \cite{acharya@sigmod1999joinsynopses} have established very early results on the hardness of sampling over joins, in the general foreign-key setting, by showing that it is generally not possible to create a uniform random sample by simply joining uniform random samples from each independent relation.
Since then, approaches have been developed to address this problem, which can be broadly grouped in the following categories: sampling, sketching, indexing, and machine learning.

\textit{Sketching-based algorithms} \cite{rusu@tdbs2008} leverage the techniques mentioned above to create sketches on the join attribute and ignoring all the other attributes.
These algorithms usually provide accurate estimates on the join size on queries without selection predicates.
In order to support queries with predicates, however,
2-way joins need to be transformed into 3-way joins
by creating a additional sketch to represent the predicate~\cite{chen@sigmod17}, which substantially deteriorates the estimation quality~\cite{vengerov@vldb2015}.

\textit{Sampling-based algorithms}  keep a sample of the tuples, apply predicates on the sample to select the tuples that satisfy them, and finally estimate the join size using only these tuples. 
Multiple tuple-sampling strategies have been proposed~\cite{acharya@sigmod1999joinsynopses, estan@icde2006endbiased, vengerov@vldb2015, chen@sigmod17}. Recent works~\cite{vengerov@vldb2015,chen@sigmod17} incorporate ideas similar to the strategy used in this paper and in KMV sketches family: they use a random hashing function to map join values to the unit range and then select tuples based on some selection strategy. For instance, the strategy adopted by the correlated sampling algorithm~\cite{vengerov@vldb2015} is equivalent to the strategy of the G-KMV sketch~\cite{yang@icde2019}, where tuples are selected if the hashed keys are smaller than a probability threshold. In contrast, \corrsketch includes tuples in the sketch up to a fixed number, which avoids assigning too much space to large datasets and leads to more predictable performance for query evaluation. 
In the future, we would like to explore trade-offs between the fixed-sized sketches we use and variable-sized sketches.


\textit{Index-assisted algorithms}~\cite{leis@cidr2017, 
lipton@sigmodrec1990, ganguly@sigmodrec1996} rely on indexes to perform the sampling. The use of indexes allows algorithms to retrieve only tuples that are relevant to the query, thus avoiding worst-case scenarios where no samples are available to perform estimations. A drawback of this approach is that indexes are not always available, and repeated index look-ups become expensive when the index does not fit into main memory.

When a join involves keys of multiple tables and highly-selective predicates, it becomes harder to estimate the join size because the join keys intersection gets increasingly smaller. Recently proposed approaches based on \textit{machine-learning} \cite{kipf@cidr2019, yang-deep-card-est@vldb2019} are able to improve the estimation performance for highly-selective worst-case queries.

%

There are key differences between our work and approaches to
join-size estimation algorithms both in terms of their goals and challenges.
Notably, \corrsketch does not need to deal with selection conditions, and one-many as wells as many-many relationships can be reduced to one-one joins (Section~\ref{sec:correlation-sketches}), allowing us to create uniform random samples of the resulting join table.

\paragraph{\Dataset Discovery and Search}
Given an input query \dataset, \dataset discovery methods have been proposed to find related \datasets that can be integrated via relational operations \textit{join} or \textit{union}.
They use the notions of Jaccard similarity and containment,  
along with variations, to compute overlaps in \dataset keys and other categorical attributes. These overlaps are 
then used to estimate how \textit{joinable} or \textit{unionable} 
two \datasets are. 
JOSIE~\cite{zhu@sigmod2019} provides an exact 
solution for finding joinable \datasets in a data lake.  LSH Ensemble~\cite{zhu@vldb2016-lsh-ensemble},  GB-KMV~\cite{yang@icde2019} and Lazo~\cite{fernandez@icde2019} propose approximate approaches to 
the same problem. 
Nargesian et al.~\cite{nargesian@vldb2018} proposed a probabilistic
solution to the problem of searching for unionable tables within massive data repositories. 
While these works focus on either finding \datasets that are joinable or union-able,
we focus on finding joinable tables that have highly correlated columns.

Several end-to-end dataset search systems have also been proposed. 
Data Civilizer~\cite{deng2017data} is a system that  uses a linkage graph to help identification of relevant data to a user task.
JUNEAU~\cite{zhang2020finding} formalized multiple data search tasks, table relatedness measures (e.g., column and row overlap, provenance, textual similarity), and proposed a linear combinations of these measures for providing data search within data science environments.
ARDA~\cite{chepurko2020arda} is a system that focuses on automatic data augmentation, i.e., how select the best features discovered from external dataset search systems.
We address an orthogonal problem: instead of building an end-to-end system, we focus on \textit{efficient} discovery and computation of a class of relationships that has been overlooked in prior works: correlation between numeric variables. Our techniques can be integrated as relationships in linkage graphs~\cite{deng2017data, fernandez@icde2018} or as a table relatedness measure~\cite{zhang2020finding} to improve search applications.

%
%

\paragraph{Correlation Estimation}
Correlation measures have been studied extensively in the literature (Section~\ref{sec:prelim_corr_estimation}).
Most works in this area focus on statistical inference, i.e., given samples of a potentially infinite population, the goal is to infer properties by deriving estimates or performing statistical tests. Recent work has focused on reducing bias and estimation errors~\cite{bishara-hittner@epm2015}, hypothesis testing~\cite{bishara@psymethods2012, dewinter@psymethods2016}, or confidence intervals (CI) that work well under non-normal distributions~\cite{bishara@bjmsp2018, bishara@brm2017, hu@tas2020}. Differently from existing CIs,  we develop a CI for sub-samples based on concentration inequalities that makes no distributional assumptions. This setting is particularly attractive for sketching algorithms that perform a single pass over the data, allowing us to leverage statistics about data (e.g., the range of the values) which would otherwise be unknown.

\section{Conclusion}
\label{sec:conclusion}

In this paper, we studied a new class of data discovery queries,  \textit{join-correlation queries}, and proposed a new approach to support them 
based on data sketches for join-correlation estimation. 
%
%
%
We also developed confidence interval bounds for the accuracy of these correlation estimates, and used them to design  functions to rank the discovered 
columns.
We experimentally evaluated both our sketches and scoring functions and showed that they are effective in a variety of \datasets and
different correlation measures. 

While we focused on correlations, our sketches are general and may be used to estimate other statistics.
Moreover, our work opens up a new avenue for the development of a new family of sketches that both join and estimate different data relationships between unjoined datasets.
For future work, we plan to explore different approaches to derive tighter confidence bounds, compare the performance and effectiveness sample selection strategies~\cite{dasgupta@icdt2016-theta-sketches, yang@icde2019}, and investigate a tighter integration between our method and others that focus on finding joinable \datasets (e.g., Lazo~\cite{fernandez@icde2018} and JOSIE~\cite{zhu@sigmod2019}). Finally, we would like to explore applications of our sketches for building end-to-end data search that takes into account the improvement in accuracy of machine learning models.

\vspace{0.5\baselineskip}
\myparagraph{Acknowledgments.}
This work was partially supported by the DARPA D3M program and NSF award OAC-1640864.
Any opinions, findings, and conclusions or recommendations expressed in this material are those of the authors and do not necessarily reflect the views of NSF and DARPA.

\vspace{0.6\baselineskip} 


\bibliographystyle{abbrv}
\bibliography{paper}

\balance

\appendix

\section{Appendix}

\subsection{Proof of Theorem \ref{thm:uniform}}

\begin{proof}[Proof of Theorem~\ref{thm:uniform}]
Let $\mathcal{T}_{X \bowtie Y} = \langle K_{X \bowtie Y}, X_{X \bowtie Y}, 
Y_{X \bowtie Y} \rangle$  be
the table resulting of the join between 
$\mathcal{T}_X = \langle K_X, X \rangle$ and $\mathcal{T}_Y = \langle K_Y, Y \rangle$. 
By definition, $K_{X \bowtie Y} = K_X \cap K_Y$. 
Let $g = h_u(h(k))$ be the composition of the hash functions $h,h_u$ described above;  $g$ maps keys from the set $K_X \cap K_Y$ uniformly at random to $[0,1]$.
 %
For this proof, let $g(K) = \{ g(k) : k \in  K\}$,  
$S_{X \bowtie Y}$ be the set of tuples $\{ \langle k, x_k, y_k\rangle : k \in K_{X \bowtie Y}\}$ with the 
$n$ smallest values $g(k) \in g(K_{X \bowtie Y})$, and $n < |K_{X \bowtie Y}|$. 
Notice that, because $g$ assigns values uniformly and randomly, the set of tuples
$\langle x_k, y_k \rangle \in S_{X \bowtie Y}$ is a uniform random sample of the set of tuples 
$\langle x_k, y_k \rangle \in\mathcal{T}_{X \bowtie Y}$.

%

Now consider the size-$n$ synopses $L_{\langle K_X,X \rangle}$ and $L_{\langle K_Y,Y \rangle}$ of 
tables $\mathcal{T}_{X}$ and $\mathcal{T}_{Y}$ respectively.
Let $L_{K_X}$ and $L_{K_Y}$ be the sets of keys from their respective synopses, 
i.e., $L_{K_X} =  \{k_x :  k_x \in L_{\langle K_X, X \rangle}\}$
and $L_{K_Y} =  \{k_y: k_y \in L_{\langle K_Y, Y \rangle}\}$.
Moreover, let $L_{X \bowtie Y } = \{ \langle k, x_k, y_k \rangle :  k \in L_{K_X} \cap L_{K_Y}\}$.
Because a synopsis $L$ always keeps the numerical values associated with their respective keys, to prove that $\langle x_k, y_k \rangle \in L_{X \bowtie Y}$ is a uniform random sample of the set of tuples $\langle x_k, y_k \rangle \in\mathcal{T}_{X \bowtie Y}$, it suffices to show that the set of keys
$\{ k:  k \in L_{X \bowtie Y} \}$ is a uniform random sample of $K_{X \bowtie Y}$.

By definition, the set of keys $L_{K_X}  \in L_{\langle K_X,X \rangle}$ 
(resp. $L_{K_Y}  \in L_{\langle K_Y,Y \rangle}$) only contains the $n$ keys $k \in K_X$ 
(resp. $k \in K_Y$) with the smallest values of $g(K_X)$ (resp. $g(K_Y)$), 
and the joined synopsis table $L_{X \bowtie Y }$ contains their intersection: 
$L_{K_X} \cap L_{K_Y}$. Thus, it is easy to see that 
$L_{K_X} \cap L_{K_Y} \subseteq \{ k : k \in S_{X \bowtie Y} \}$ always holds.
%
Without loss of generality, assume that $|L_{\langle K_X,X \rangle} | = | L_{\langle K_Y,Y \rangle}| = |S_{X \bowtie Y}| = n$.
The best case happens when the sets of keys are equal, i.e., $L_{K_X} = L_{K_Y}$, in which case $| L_{K_X} \cap L_{K_Y} | = |S_{X \bowtie Y}| $ and  
$ L_{K_X} \cap L_{K_Y} =  \{ k : k \in S_{X \bowtie Y} \}$. 
When  $L_{K_X} \neq L_{K_Y}$, then $| L_{K_X} \cap L_{K_Y} | < |S_{X \bowtie Y}| $. 
Now, assume that $| L_{K_X} \cap L_{K_Y} | = 1 $. 
Then, the single key $k \in L_{K_X} \cap L_{K_Y}$ has the smallest value of $g(k)$.
The sample $S_{X \bowtie Y}$ of size $1$ also contains the same key $k \in L_{K_X} \cap L_{K_Y}$.
More generally, if $| L_{K_X} \cap L_{K_Y} | = m $,  then the set of keys of the sample 
$S_{X \bowtie Y}$ of size $m$ is equal to the set $L_{K_X} \cap L_{K_Y}$.
Therefore, the set of tuples $\{ \langle x_k, y_k \rangle : k \in L_{K_X} \cap L_{K_Y} \}$ induced by $L_{X \bowtie Y}$ is also a uniform random sample of $\mathcal{T}_{X \bowtie Y}$.
\end{proof}

\end{document}